%% file: icrc-wallet.tex
\documentclass[conference,compsoc]{IEEEtran}

\input{preambule}

\begin{document}
%
\title{A Low-Cost Privacy-Preserving Digital Wallet for Humanitarian Aid Distribution}


\author{
    \IEEEauthorblockN{Eva Luvison\IEEEauthorrefmark{1,3}, Sylvain Chatel\IEEEauthorrefmark{1}, Justinas Sukaitis\IEEEauthorrefmark{2}, Vincent Graf Narbel\IEEEauthorrefmark{2}, Carmela Troncoso\IEEEauthorrefmark{3}, Wouter Lueks\IEEEauthorrefmark{1}}
    \IEEEauthorblockA{\IEEEauthorrefmark{1}CISPA Helmholtz Center for Information Security, Saarbrücken, Germany
    \\\{eva.luvison, sylvain.chatel, lueks\}@cispa.de}
    \IEEEauthorblockA{\IEEEauthorrefmark{2}International Committee of the Red Cross, Geneva, Switzerland
    \\dpo@icrc.org}
    \IEEEauthorblockA{\IEEEauthorrefmark{3} SPRING Lab, EPFL, Laussanne, Switzerland
    \\carmela.troncoso@epfl.ch}
}

\maketitle

\begin{abstract}
Humanitarian organizations distribute aid to people affected by armed conflicts or natural disasters. Digitalization has the potential to increase the efficiency and fairness of aid-distribution systems, and recent work by Wang et al. has shown that these benefits are possible without creating privacy harms for aid recipients. However, their work only provides a solution for one particular aid-distribution scenario in which aid recipients receive a pre-defined set of goods. Yet, in many situations it is desirable to enable recipients to decide which items they need at each moment to satisfy their specific needs. We formalize these needs into functional, deployment, security, and privacy requirements, and design a privacy-preserving digital wallet for aid distribution. Our smart-card-based solution enables aid recipients to spend a pre-defined budget at different vendors to obtain the items that they need. We prove our solution's security and privacy properties, and show it is practical at scale.
\end{abstract}

\section{Introduction}
\label{sec:introduction}
\input{parts/1-introduction}

\section{Problem and Solution Overview}
\label{sec:pb}
\input{parts/2-problemDef}

\section{Our Solution}
\label{sec:sol}
\input{parts/3-solution}

\section{Security and Privacy Analsis}
\label{sec:sec}
\input{parts/4-secuAndPrivacy}

\section{Deployment in Practice}
\label{sec:disc}
\input{parts/6-discussion}

\section{Evaluation}
\label{sec:eval}
\input{parts/5-evaluation}

\section{Conclusion}
\label{sec:conc}
\input{parts/7-conclusion}

\section*{Acknowledgment}
We would like to thank Kevin Gni for helping with the prototype code and evaluation.

Initial work for this paper was funded by the Science and
Technology for Humanitarian Action Challenges (HAC) programme from the Engineering for Humanitarian Action
initiative, a partnership between the ICRC, EPFL and ETHZ.

\bibliographystyle{IEEEtranS}
\bibliography{literature}

\appendices

\input{parts/appendix}

\end{document}

%% file: preambule.tex
\usepackage{tikz}
\usepackage{amsmath}

\PassOptionsToPackage{hyphens}{url}
\usepackage{hyperref}

\usepackage{mathtools} 
\usepackage{amsthm}
\usepackage{enumitem}
\usepackage{bbold}

\usepackage[scaled=0.88]{helvet}

\usepackage{subcaption}

\usepackage{array}
\newcolumntype{L}[1]{>{\raggedright\let\newline\\\arraybackslash\hspace{0pt}}m{#1}}
\newcolumntype{C}[1]{>{\centering\let\newline\\\arraybackslash\hspace{0pt}}m{#1}}
\usepackage{multirow}
\usepackage{rotating}
\newcolumntype{R}[2]{%
    >{\adjustbox{angle=#1,lap=\width-(#2)}\bgroup}%
    l%
    <{\egroup}%
}
\usepackage{xcolor,colortbl}

\usepackage[n, advantage, operators, sets, adversary, landau, probability, notions, logic, ff, mm, primitives, events, complexity, oracles, asymptotics, keys]{cryptocode}

\theoremstyle{definition}

\newtheorem{definition}{Definition}

\newcommand{\descr}[1]{\smallskip \noindent \textbf{#1}}

\newcommand{\descrit}[1]{\smallskip \noindent \textit{#1}}
\newcommand{\etal}{et al.\xspace}
\newcommand{\ie}{\text{i.e.,\,}}
\newcommand{\eg}{\text{e.g.,\,}}

\newcommand{\ICRC}{\text{ICRC}}

\newcounter{statement}

\graphicspath{ {./images/} }

\newtheorem{theorem}{Theorem}

\theoremstyle{remark}

\def\addvalue#1#2{\expandafter\gdef\csname my@data@#1\endcsname{#2}}%
\def\usevalue#1{\csname my@data@#1\endcsname}

\newcommand{\reqdef}[3]{%
\addvalue{#2}{RQ.#1}%
\hypertarget{#2}\noindent\textit{\usevalue{#2}: #3.}
}
\newcommand{\reqlinky}[1]{\hyperlink{#1}{\usevalue{#1}}}

\usepackage{algorithm,algpseudocodex}

\newcommand{\code}[1]{\textsf{#1}}

\usepackage{xspace}

\newcommand{\AEnc}{\code{AE}\xspace}
\newcommand{\KGen}{\code{KGen}}
\newcommand{\Enc}{\code{Enc}}
\newcommand{\Dec}{\code{Dec}}

\newcommand{\DS}{\code{DS}\xspace}
\newcommand{\Sign}{\code{Sign}}
\newcommand{\Verif}{\code{Verif}}

\newcommand{\Com}{\code{Com}\xspace}
\newcommand{\Gen}{\code{Gen}}
\newcommand{\Commit}{\code{Commit}}

\newcommand{\MAC}{\code{MAC}\xspace}
\newcommand{\Tag}{\code{Tag}}

\newcommand{\ORAM}{\code{ORAM}\xspace}
\newcommand{\Init}{\code{Init}}
\newcommand{\Read}{\code{Read}}
\newcommand{\Write}{\code{Write}}
\newcommand{\Serve}{\code{Serve}}

\newcommand{\mess}{\ensuremath{m}\xspace}
\newcommand{\cipher}{$c$\xspace}
\newcommand{\N}{$N$\xspace}

\newcommand{\spend}{\ensuremath{\code{spent}}\xspace}
\newcommand{\spendseen}{\ensuremath{\code{spent}_{\code{seen}}}}
\newcommand{\spendmax}{\ensuremath{\code{spent}_{\code{max}}}}

\newcommand{\spendsumsim}{\ensuremath{\code{spent}_{\code{sum}}}\xspace}
\newcommand{\spendsum}{\ensuremath{\code{spent}^{*}_{\code{sum}}}\xspace}

\newcommand{\spendsumz}{\ensuremath{\code{spent}^{*0}_{\code{sum}}}\xspace}
\newcommand{\spendsumo}{\ensuremath{\code{spent}^{*1}_{\code{sum}}}\xspace}

\newcommand{\valid}{\ensuremath{\code{valid}}}

\newcommand{\myb}{\ensuremath{\code{b}}}

\newcommand{\myst}{\ensuremath{\code{st}}\xspace}
\newcommand{\success}{\ensuremath{\code{success}}\xspace}
\newcommand{\db}{\ensuremath{\code{DB}}\xspace}
\newcommand{\price}{\ensuremath{\code{price}}\xspace}
\newcommand{\amountrec}{\ensuremath{\code{amount\_received}}\xspace}
\newcommand{\amountspent}{\ensuremath{\code{amount\_spent}}\xspace}
\newcommand{\world}{\ensuremath{\code{world}}\xspace}
\newcommand{\bud}{\ensuremath{\code{bud}}\xspace}
\newcommand{\ctr}{\ensuremath{\code{ctr}}\xspace}

\newcommand{\Hid}{\ensuremath{\code{id}_{\code{H}}}\xspace}

\newcommand{\tid}{\ensuremath{\code{t}_{\code{id}}}\xspace}
\newcommand{\tidX}[1]{\ensuremath{\code{t}_{\code{id}_{#1}}\xspace}}
\newcommand{\tnb}{\ensuremath{\code{t}_{\code{nb}}}\xspace}
\newcommand{\aggproof}{\ensuremath{\pi_{\code{recl}}}\xspace}

\newcommand{\bal}{\ensuremath{\code{balance}}\xspace}

\newcommand{\counter}{\ensuremath{\code{counter}}\xspace}

\newcommand{\myM}{\ensuremath{\mathcal{M}}\xspace}
\newcommand{\myTH}{\ensuremath{\mathcal{T}_H}\xspace}

\newcommand{\inlog}{\ensuremath{\code{ReceivedTrans}}}
\newcommand{\spenttrans}{\ensuremath{\code{SpentTrans}}}
\newcommand{\Bud}{\ensuremath{\code{Bud}}}

\newcommand{\SMap}{\ensuremath{\code{CardSMap}}}

\newcommand{\prfk}{\code{K}}

\newcommand{\oraclesym}{\ensuremath{\mathcal{O}}}

\newcommand{\skt}{{\color{violet}{\ensuremath{\sk\textsubscript{T}}}}\xspace}
\newcommand{\pkt}{{\color{violet}{\ensuremath{\pk\textsubscript{T}}}}\xspace}
\newcommand{\sks}{{\color{orange}{\ensuremath{\sk\textsubscript{RS}}}}\xspace}
\newcommand{\pks}{{\color{orange}{\ensuremath{\pk_{\textrm{RS}}}}\xspace}}
\newcommand{\sksig}{{{\ensuremath{\sk}}}\xspace}
\newcommand{\pksig}{{{\ensuremath{\pk}}}\xspace}

\newcommand{\larrow}{\ensuremath{\leftarrow}\xspace}
\newcommand{\secuparam}{\ensuremath{1^\ell}\xspace}
\newcommand{\params}{\ensuremath{\code{params}}\xspace}

\newcommand{\oraclehreg}{\ensuremath{\oraclesym_{\textrm{HReg}}(\sks, \cdot)}\xspace}
\newcommand{\oraclemureg}{\ensuremath{\oraclesym_{\textrm{MUserReg}}(\sks, \cdot)}\xspace}
\newcommand{\oraclemsreg}{\ensuremath{\oraclesym_{\textrm{CStationReg}}(\cdot)}\xspace}
\newcommand{\oracleregtwo}{\ensuremath{\oraclesym_{\textrm{RegSplitWorld}}(\cdot, \sks, \cdot)}\xspace}

\newcommand{\oraclespend}{\ensuremath{\oraclesym_{\textrm{Spend}}(\cdot)}\xspace}
\newcommand{\oraclespendmu}{\ensuremath{\oraclesym_{\textrm{SpendMalUser}}(\cdot)}\xspace}
\newcommand{\oraclespendmv}{\ensuremath{\oraclesym_{\textrm{SpendMalVendor}}(\cdot)}\xspace}
\newcommand{\oraclespendtwo}{\ensuremath{\oraclesym_{\textrm{SpendSplitWorld}}(\cdot)}\xspace}

\newcommand{\oraclehregargs}{\ensuremath{\oraclesym_{\textrm{HReg}}}\xspace}
\newcommand{\oraclemuregargs}{\ensuremath{\oraclesym_{\textrm{MUserReg}}}\xspace}
\newcommand{\oraclemsregargs}{\ensuremath{\oraclesym_{\textrm{CStationReg}}}\xspace}
\newcommand{\oracleregtwoargs}{\ensuremath{\oraclesym_{\textrm{RegSplitWorld}}}\xspace}

\newcommand{\oraclespendargs}{\ensuremath{\oraclesym_{\textrm{Spend}}}\xspace}
\newcommand{\oraclespendmuargs}{\ensuremath{\oraclesym_{\textrm{SpendMalUser}}}\xspace}
\newcommand{\oraclespendmvargs}{\ensuremath{\oraclesym_{\textrm{SpendMalVendor}}}\xspace}
\newcommand{\oraclespendtwoargs}{\ensuremath{\oraclesym_{\textrm{SpendSplitWorld}}}\xspace}

\newcommand{\expbase}[1]{\ensuremath{\code{Exp}^{\text{#1}}_{\adv}}\xspace}
\newcommand{\expbaseb}[2]{\ensuremath{\code{Exp}^{\text{#1}}_{\adv, #2}}\xspace}
\newcommand{\expind}[1][b]{\expbaseb{IND}{#1}}
\newcommand{\expsec}{\expbase{SEC}}

\newcommand{\expred}{\expbase{RECL}}

\newcommand{\expaud}[1][b]{\expbaseb{AUDP}{#1}}

\newcommand{\Spend}{\ensuremath{\code{Spend}}\xspace}
\newcommand{\Receive}{\ensuremath{\code{Receive}}}

\newcommand{\GenProof}{\ensuremath{\code{CreateReclaimProof}}}
\newcommand{\VerifyProof}{\ensuremath{\code{VerifyReclaimProof}}}

\newcommand{\VerifyRedProof}{\ensuremath{\code{VerifyRedeemProof}}}

\newcommand{\Request}{\ensuremath{\code{Request}}}
\newcommand{\Allocate}{\ensuremath{\code{Allocate}}}

\newcommand{\TrustedSetup}{\ensuremath{\code{TrustedSetup}}}
\newcommand{\SetupSign}{\ensuremath{\code{SetupRSKeys}}}
\newcommand{\SetupToken}{\ensuremath{\code{SetupCard}}}

%% file: parts/1-introduction.tex
Humanitarian organizations provide assistance to people affected by extreme circumstances such as conflicts, famine, or natural disasters. One core aspect of assistance is to distribute aid in the form of food and other essential items~\cite{icrc2019ecosec}. While traditionally these distribution processes were analog, humanitarian organizations are increasingly considering the adoption of digital technologies to increase efficiency and accountability, which in turn enables them to reach as many people as possible with limited resources.

The introduction of digital technologies, however, brings with it the potential for harm to recipients of aid~\cite{KaspersenICRC16}. Since recipients have very little choice in refusing aid, it is essential that aid distribution digitalization is designed to prevent such harms~\cite{BurtonICRC20}. Unfortunately, existing technological solutions often do not protect the privacy of recipients, and the collection of personal data can introduce risks for recipients such as prosecution~\cite{ciesielski2022afghans}, or losing citizenship~\cite{aljazeera2019kenya}. 

Wang \etal~\cite{wang2023digital} were the first to identify the challenges of aid distribution digitalization. The deployment constraints make designing a (privacy-preserving) solution difficult. First, aid distribution is usually assigned \emph{per household} requiring multiple household members to be able to access a shared amount of aid. Second, solutions must often be low-tech: recipients cannot always be assumed to have high-performance hardware, and the existence of digital communication between parties is not guaranteed. To address these challenges, Wang \etal proposed a token-based system where household members receive a smart card or use a smartphone app. Once per distribution round, one of the household members can then use their card (or phone) to obtain a fixed aid package.

Our interactions with the International Committee of the Red Cross (ICRC; a large humanitarian organization) revealed that distributing fixed aid packages is not always the preferred choice. Fixed-package distribution can be deployed quickly and therefore works well as a stop-gap solution in emergency situations. In other situations, however, humanitarian organizations prefer to assign recipients a budget that they can then \emph{flexibly} spend at vendors of their choice. This has two advantages. First, it increases the agency of aid recipients by letting them procure the items \emph{they know they need}, rather than giving them the items the organization has decided they need. Second, letting recipients obtain goods from local vendors supports the creation of a local economy and ecosystem.

The obvious solution of simply giving cash to aid recipients often does not suffice. Distributing cash is time-consuming and error-prone~\cite{ifrc2021identity}. Additionally, transporting and storing large amounts of cash -- as it is needed to support potentially hundreds of families for long periods -- in conflict zones puts humanitarian personnel in danger. Hence, there is a need for a digital wallet solution that operates independently of existing monetary infrastructures.

We collaborated closely with the ICRC to better understand the requirements of a \emph{humanitarian-oriented digital wallet}. Such wallets would be used in mid- or long-term aid distribution campaigns in one of two scenarios. Either they are deployed in a setting where recipients have access to high-end devices and good connectivity -- such as the recent armed conflict in Ukraine; or recipients have no guaranteed access to programmable powerful devices, or scarce connectivity -- such as conflicts in sub-Saharan Africa. In this paper, we focus on the second scenario, in which recipients cannot use their own devices and humanitarian organizations must provide them with a digital support for their wallets. To keep the cost reasonable, we choose as support \textit{smart cards}. Designing digital wallets based on smart cards is challenging in the aid distribution scenario because \emph{multiple smart cards} must be able to access the same household budget, yet these cards cannot directly communicate with each other to share the latest balance. In our solution we use techniques from oblivious memory access to hide which card (or household) is making transactions.

The synchronization and privacy requirements also rule out many other existing wallet-like solutions. Designs based on blockchain and cryptocurrencies do not provide the expected privacy requirements~\cite{RinbergA22SoKanonymity}; depend too much on high connectivity between tokens; require high-tech technologies or heavy computations~\cite{AlmashaqbehS22SoKblockchain} to be deployable; or in the case of e-cash still require privately synchronizing which coins have been spent. The custom solution of Wang \etal~\cite{wang2023digital} on the other hand is privacy-friendly and deployable, but cannot be converted into a wallet-based system: the assumption that in every round households always withdraw the entire aid budget in full at a single distribution point is baked into their design.

In this paper, we make the following contributions:
\begin{itemize}[leftmargin=*]
\item We present the functional, security, and privacy requirements as well as deployment constraints of wallet-based aid distribution in humanitarian scenarios when high-end user devices are not available. One particularly challenging constraint is that recipient-held smart cards cannot communicate with each other.
\item We present a smart card-based wallet system that addresses these security, privacy, and functionality requirements. Tokens do not communicate with each other, but instead, privately synchronize their state via the vendors without harming recipient privacy.
\item We formally model the security and privacy properties that a wallet-based aid distribution system should satisfy. In particular, we model how privacy is maintained even though tokens of different users might now hold different balances. We prove that our new proposal satisfies all security and privacy requirements.
\item We implemented a prototype of our system to show that our designs are practical for medium-scale aid distribution projects.
\end{itemize}

%% file: parts/2-problemDef.tex
We first present our system and threat models and describe the problem. Then, we sketch an overview of our solution and discuss why previous works cannot address the problem.

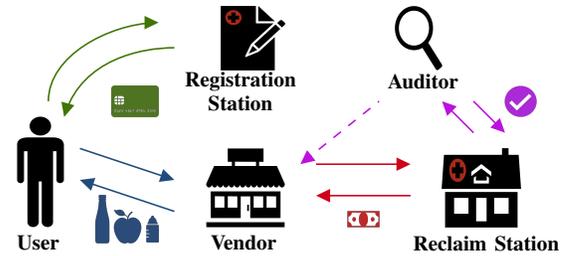
\begin{figure}[t]
    \centering
    \resizebox{0.9\columnwidth}{!}{
    \input{images/little_overview.tikz}
    }
    \captionsetup{font=small} 
    \caption{Workflow of our digital wallet. The user interacts with the registration station which allocates a budget and returns a smart card. The user transacts with a vendor who later contacts the reclaim station to reclaim the transaction. An auditor can check the correctness of this reclaim.}
    \label{fig:small-overview}
\end{figure}

\subsection{System Model}\label{sec:pb:model}
We consider an aid distribution scenario shown in Fig~\ref{fig:small-overview}. As in the work of Wang \etal~\cite{wang2023digital}, aid recipients are part of a household, and each household is allocated a specific budget that limits what recipients in that household can spend. Households register with the registration station that determines eligibility and allocates their budget. At the end of the registration process, a subset of household members is handed a smart card, which provides tamper-resistant hardware and secure storage~\cite{ShepherdAGLMASC16}.

Contrary to the work by Wang \etal, however, recipients can spend their budget at different vendors. And they do not have to spend their budget all in one go. Vendors are typically independent and are not operated by a humanitarian organization. After one or more transactions, a vendor can reclaim the money by interacting with the reclaim station operated by the humanitarian organization. We follow Wang \etal's requirement for auditability to ensure humanitarian funds are well spent. In our system, however, the auditor audits the reclaim station, and thereby indirectly the vendors.

\descr{Key challenge: household-shared budgets.} All household members can access the same budget in the same way that several bank cards can access the same shared bank account. Such centralized banking systems, however, do not satisfy the \ICRC's strict privacy needs. In particular, due to the requirements of financial, anti-terrorist, and anti-fraud regulations, banks need to collect a large amount of information to check that customers are not on government lists, and if they are, report them.
The requirement of sharing a budget together with the requirement that tokens do \emph{not} communicate rules out many other obvious solutions. Electronic-cash-based solutions and privacy-friendly distributed ledgers require tokens to be able to synchronize state to know which coins can still be spent. Pre-paid cards and analog vouchers cannot be duplicated or shared among household members. And, while digital vouchers are easy to copy and share, this approach results in privacy harm if two members spend the same voucher.

\descr{Parties.}
Our system involves several actors:

    \descrit{Trusted party:} that builds and initiates the tokens. 
    
    \descrit{Users:} that receive aid. They are grouped by household. A user is \textit{legitimate} if the \ICRC\ registered and allocated a budget to their household. Such users will be issued a personal token. Otherwise, they are \textit{illegitimate} and cannot register themselves. A user can belong to one and only one household. 
    
    \descrit{Registration Station:} that enrolls the legitimate users and allocates a budget to their households. The registration station is usually run by the \ICRC\ staff.
    
    \descrit{Vendors:} that are approved by the \ICRC\ are able to sell goods to legitimate users via our system.
    
    \descrit{Reclaim Station:} that the vendors interact with to reclaim the amount for which they have sold goods. The reclaim station uses transaction records provided by the vendor to verify the amount. The reclaim station is run by the \ICRC\ staff.
    
    \descrit{Auditors:} that can audit the records of the reclaim station to verify that the amount of money the reclaim station paid to vendors corresponds to transactions by legitimate users. They use the transaction records.

\descr{Threat Model}
We consider that the trusted party is always honest. For the other actors, their capabilities depend on their role in the protocol. 

We assume that the users are malicious (i.e., they might cheat when registering or transacting with the vendor). Yet, they still want to keep their privacy.

We assume that the registration station, the reclaim station, and vendors are honest with respect to the operations inherent to their roles in the system: the registration station honestly registers users, the reclaim station lets legitimate vendors reclaim their money; and vendors honestly hand over the goods that recipients paid for.

Vendors and the reclaim station can be malicious with respect to money. Vendors might collude with users to try to reclaim more money than what they transacted. The reclaim station and vendors might try to convince an auditor that money went to legitimate recipients when it did not.

With respect to privacy, vendors, the registration station, the reclaim station, and auditors can behave adversarially to break the user's privacy (i.e., to obtain details about transactions; or to distinguish between users). Since cards cannot share state, malicious vendors could execute roll-back attacks on the database to try to distinguish users. We therefore aim to achieve privacy against covert parties only.

Throughout this paper, we assume that tokens communicate with the registration station and vendors through authenticated communication channels~\cite{UngerDBFPG015SoKmessaging}. 

\subsection{Requirements}\label{sec:pb:req}
We now lay down the requirements for a secure and privacy-preserving wallet-based aid distribution digital system. We build on top of prior work for fixed-distribution systems~\cite{wang2023digital} but refine all requirements to ensure compatibility with a wallet-based approach. For completeness, we list all requirements here but highlight (using ***) new and significantly changed requirements.

\subsubsection{Functional requirements}
The system must have the following functionalities:

\reqdef{F1}{house}{Allocation per household} Aid can be given per individual or per household. In this work, we focus on a per-household allocation of the budget.

\reqdef{F2}{budget}{Budget allocation***} Households receive a budget that they can spend at different vendors. Here we differ from Wang et al.~\cite{wang2023digital} who instead assume a fixed periodic entitlement that is obtained from one distribution station.

\reqdef{F3}{period}{Periodical budget} The budget is allocated per period. At the beginning of a new period, the balance is updated to the initial allocated budget. This is done automatically without the need to register again. 
Without loss of generality, the remaining of the paper focuses on a single period and we present in Section~\ref{sec:disc} how to extend our approach to multiple periods. 

\reqdef{F4}{edit}{Editable tokens} As a household's situation can vary (e.g., deaths or births) over time,  the \ICRC\ must be able to modify the entitlement of a household. 

\reqdef{F5}{Vreceipt}{Accountability for vendor***} The vendor needs to keep track of their sales for their own accounting. Hence, the system should print an invoice for the vendor at each purchase.

\reqdef{F6}{proof}{Proof of sales***} To provide a vendor with the reimbursement they ask for, the reclaim station needs proof that the vendor sold goods for that amount of money. We assume that vendors can only reclaim funds during a fixed amount of time called the \textit{reclaim period}. 

\subsubsection{Deployment requirements}
The system must respect the following field constraints:

\reqdef{D1}{robust}{Robustness}
Being able to receive aid is essential for households. Therefore we follow Wang et al.~\cite{wang2023digital} in requiring robustness. Concretely, households should still be able to use their budget even when the head of household is unavailable, or when one of the households member's smart cards has been lost or damaged. To ensure availability of aid in these cases, we require that households can obtain more than one smart card, all of which should be able to access the same budget.

\reqdef{D2}{limited}{Constrained hardware***} 
In this paper, we focus on the setting where users do not (all) have (smart) phones, and households must thus rely on smart cards that they have been issued. Smart cards have only limited computation (ruling out heavy cryptography such as SNARKS, FHE, and zero-knowledge proofs) and communication capabilities; and in particular, cannot communicate with each other.

We do assume vendors and registration stations can communicate. However, smart cards of household members cannot benefit from this infrastructure to relay data between cards to enable them to synchronize state between them directly: It is very unlikely that they are connected at the same time. Instead, the challenge is to leverage this communication infrastructure to build an \emph{asynchronous} synchronization method, which is what our solution achieves.

\reqdef{D3}{scalability}{Scalability}
Humanitarian organizations organize aid in programs that range from assisting a few hundred households to assisting hundreds of thousands of households. The system should scale to be efficient for all these sizes.

\subsubsection{Security requirements}
Digital aid distribution systems should help ensure the funds allocated for aid are reaching the eligible households. We capture this in the following requirements.

\reqdef{S1}{overspending}{Overspending prevention***} Users should not be able to spend more than the allocated budget for their household. This requirement is particularly challenging as different household members share one budget, thus requiring synchronization. When users or their households were never registered, they should not be able to spend any budget.

\reqdef{S3}{over-reclaim}{Vendor over-reclaim***} The vendor should not be able to lie about the amount of goods they have delivered and reclaim a higher total. 

\reqdef{S4}{auditsec}{Auditability***} Humanitarian organizations' budget relies heavily on external donors. For these donations to continue, organizations need to prove they used these donations appropriately. Therefore, auditors should be able to review and verify the distribution process expenses. In particular, the auditor should be able to verify that the amount of money reclaimed by vendors from the reclaim station corresponds to the total amount of legitimate transactions at these vendors.

\subsubsection{Privacy requirements}
Finally, users' privacy should be protected.

\reqdef{P1}{unlinkability}{Privacy at purchase}The vendor should not learn anything about the user: name, other people in their household, allocated budget, remaining balance, previous purchase, etc. The vendor can only learn the legitimacy of the user and whether they have enough budget left to buy the goods. In particular, transactions should be unlinkable even when households spend their budget in multiple transactions. We assume the registration station is honest but curious.

As we explain in Section~\ref{sec:privacy-unlink}, without any communication between a household's cards, the best we can hope for is to achieve this property against a covert adversary. A fully malicious adversary can rewind the central balance store, thus making it possible to distinguish users. But cards can detect this rewinding, see Section~\ref{sec:disc}.

\reqdef{P2}{reclaim}{Privacy at reclaim} The reclaim station should not learn anything about the users except the total amount that the vendor is reclaiming, and whether this amount is the total sum of transactions performed by registered household members.

\reqdef{P3}{audit}{Privacy at auditing} Similarly, the auditor should not learn anything about the users, except the sum of all legitimate transactions.

\subsection{Solution Sketch}\label{sec:pb:solsketch}
We now describe an overview of our system. We distinguish five phases: setup, registration, transaction, reclaim, and audit. Figure~\ref{fig:overview} shows an overview of these phases. We refer to Section~\ref{sec:sol:syntax} for a detailed description of the syntax.

\descr{Setup.} The trusted party generates the cryptographic material for an encrypted database of per-household balances. Every smart card is initialized with the private key material to access this database.\footnote{The secure storage property of the cards ensures that this key remains private.} The registration station and vendors can access this database.

\descr{Registration.} The household members who wish to receive a smart card go to the registration station at the same time. When the household is eligible, the station assigns the household a budget and produces a smart card for each member. One of the tokens writes the initial budget to the database.

\descr{Transaction.} To perform a transaction with a registered vendor, the user first unlocks it. The card then interacts with the database and obliviously retrieves the household's balance. If the balance is insufficient, the process aborts. Otherwise, the token obliviously writes the new balance to the database and generates a \textit{proof of payment}. The proof contains a homomorphic commitment to the amount spent and is signed by the card.

\descr{Reclaim.} Periodically, the vendor can reclaim the funds for the goods they provided to users. The vendor contacts the reclaim station and using a similar technique as in Wang et al.~\cite{wang2023digital} proves that the amount is correct. The vendor submits all proofs of payments, and opens the homomorphic sum of the individual commitments. The reclaim station can check that the sum and the individual transactions are correct; if they are, the reclaim station refunds the total amount to the vendor.

\descr{Audit.} The reclaim station, in turn, can relay the proofs of payments of all vendors to convince an auditor that the amount of money reclaimed by vendors corresponds to legitimate transactions by members of eligible households.

\subsection{Related Work}\label{sec:pb:rw}

The topic of digital currencies has been a prolific line of research for the last decades. 
The emergence of scalable and efficient constructions enabled the development of different approaches such as e-cash, blockchains, and digital banks. Yet, none of these solutions match the stringent requirements for aid distribution laid out by the \ICRC. 

E-cash was introduced by David Chaum to enable privacy through untraceable payments~\cite{Chaum82ecash}. It enables clients to withdraw signed electronic coins from a bank and spend them at merchants. The merchant then checks with the bank if the transaction is valid or not. 
This protocol comes with two main issues: (i)~it uses computationally heavy cryptography (\eg RSA-based blind signatures) and (ii)~it needs to check for double-spending at each coin transaction.
Several extensions were proposed to enable additional guarantees such as additional privacy guarantees~\cite{ChaumFN88ecash}, regulation and accountability options~\cite{camenisch2006balancing}, transferrable coins~\cite{DBLP:conf/pkc/BaldimtsiCFK15, DBLP:conf/acns/CanardG08}, and memory/communication optimizations~\cite{camenisch2005compact, DBLP:conf/pkc/CanardPST15}. Unfortunately, all these solutions are still impractical for our use case which requires the use of low-tech hardware and low communication cost. On top of the high computational cost of e-cash, in our setting, different tokens of a household (see \reqlinky{robust}) must share the state of the household budget (\ie the current set of unspent coins/balance). To maintain privacy (\reqlinky{unlinkability}) this sharing needs to be done obliviously. To minimize sharing costs, we directly share the available balance (only 3 bytes) instead of sharing more information about coins.

Blockchain-based solutions such as Zerocash~\cite{DBLP:conf/sp/Ben-SassonCG0MTV14} and Zerocoin~\cite{MiersG0R13} offer an account-based alternative to the token-based digital currencies. Such solutions rely on distributed ledgers to enable private and unlinkable transactions. Yet, such approaches have a large computing and communication cost for the clients and can have scalability issues: e.g., every transaction has to be downloaded and decrypted by the clients, clients need to construct complicated zero-knowledge proofs, and clients still need to share the new coins containing the remaining balances (which are 100s of bytes big). Central bank digital currencies (CBDC) construction can combine both approaches for efficiency and functionality purposes~\cite{DBLP:conf/ndss/DanezisM16, DBLP:conf/fc/WustKCC19, DBLP:conf/ccs/WustKDC22, DBLP:conf/ccs/KiayiasKS22}. 

Neither of these two approaches is concretely efficient enough for resource-constrained devices (\reqlinky{limited}) because spending is expensive. Related work reports $2$min for Zcash~\cite{DBLP:conf/sp/Ben-SassonCG0MTV14} on an Intel i7 and $0.8$s for Platypus~\cite{DBLP:conf/ccs/WustKDC22} on an iPhone 13 mini. Since smart cards are orders of magnitude slower, we estimate run times of several minutes on a smart card, assuming access to non-commodity cards that provide support for cryptographic building blocks such as pairings needed for Zcash and Platypus. Instead, we only require a single standard digital signature (ECDSA) per transaction.

Closer to our work about the digitalization of humanitarian organization processes, two projects stand out for aid distribution in crisis zones: the Partisia project~\cite{icrc2023partisia,goldie2022partisia} and the Building Blocks project~\cite{wfp2023buildingblocks}. Unfortunately, documentation on these projects is scarce. These projects use smartphones and blockchain technologies. 
As such, these projects consider different system models compared to ours making comparisons difficult. 
Wang \etal recently proposed a digital aid-distribution system~\cite{wang2023digital}. Yet, their solution does not offer the flexibility provided by our digital wallet. In their work, the user can only retrieve a fixed amount of goods in one transaction. 
Additionally, their work can only detect double-spending by different household members, but it cannot prevent it. Instead, they implicitly assume household members can synchronize about whether they already picked up this month's aid or not. This simplification makes a lot of sense in the fixed-distribution case but does not when spending parts of budgets at different vendors.

%% file: images/little_overview.tikz
\tikzset{every picture/.style={line width=0.75pt}} 

\begin{tikzpicture}[x=0.75pt,y=0.75pt,yscale=-1,xscale=1]

\draw (155,215) node  {\includegraphics[width=22.5pt,height=52.5pt]{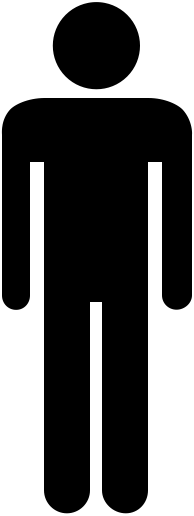}};
\draw (290,130) node  {\includegraphics[width=30pt,height=30pt]{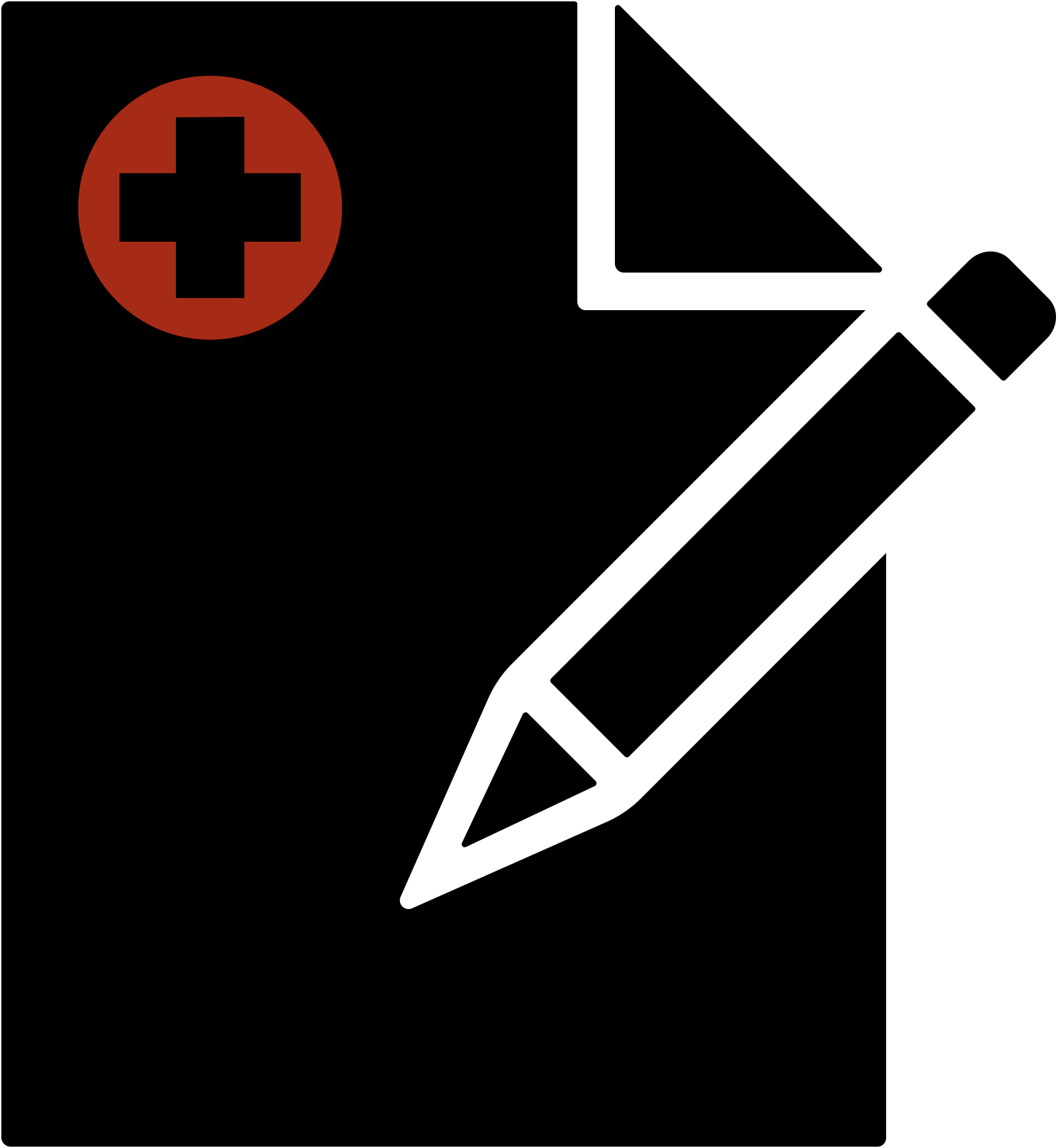}};
\draw (285,225) node  {\includegraphics[width=37.5pt,height=37.5pt]{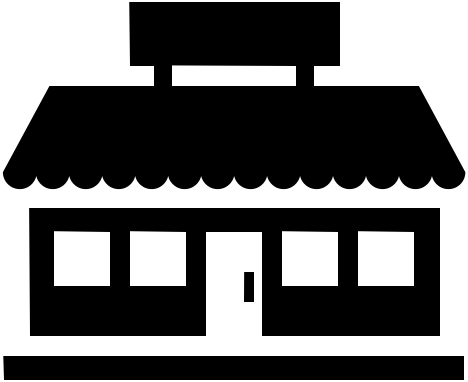}};
\draw (435,225) node  {\includegraphics[width=37.5pt,height=37.5pt]{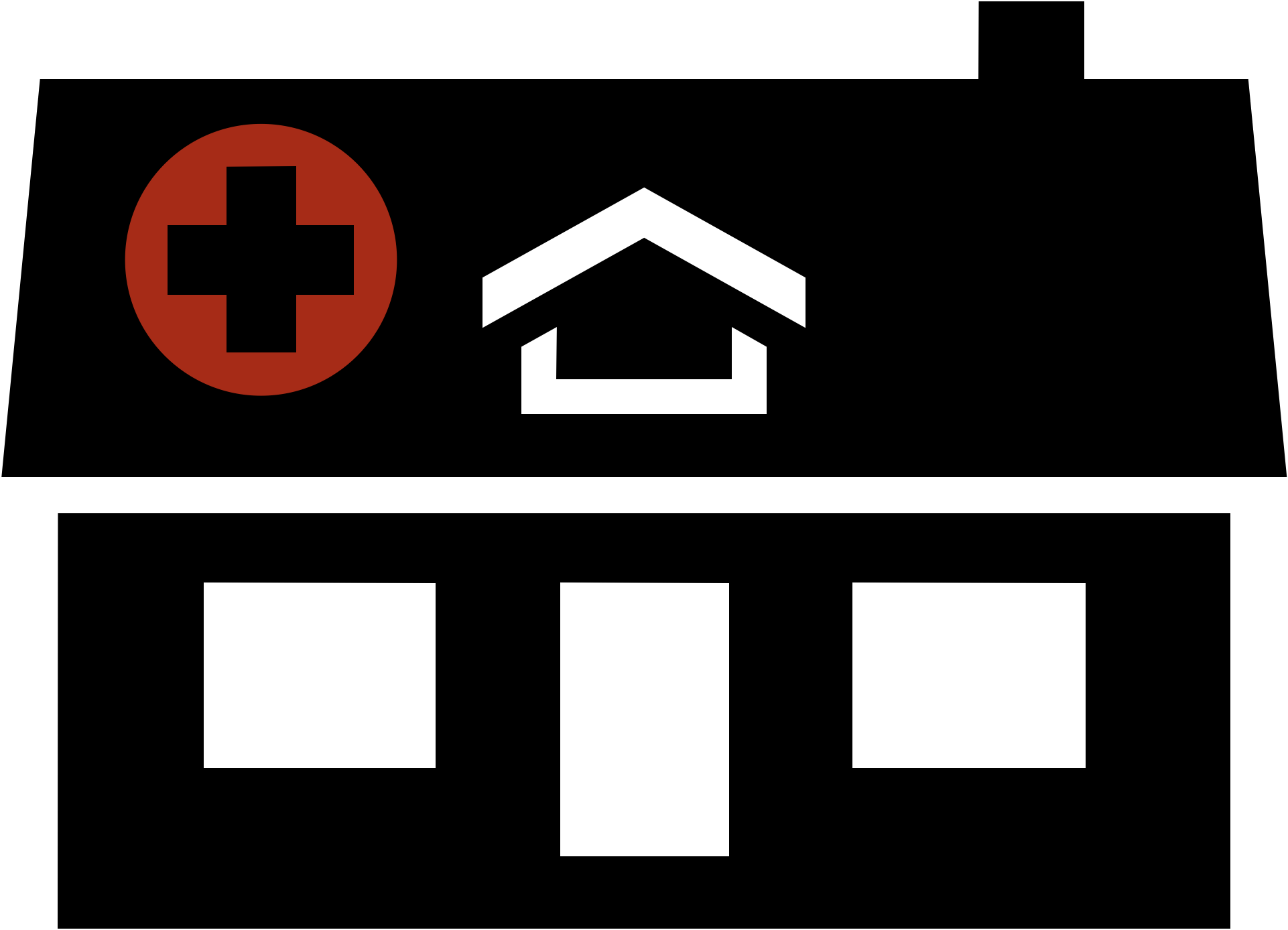}};
\draw (395,130) node  {\includegraphics[width=22.5pt,height=30pt]{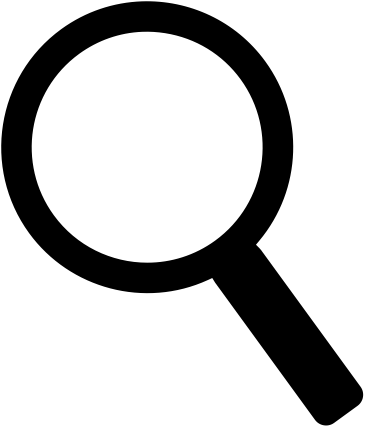}};
\draw [color={rgb, 255:red, 65; green, 117; blue, 5 }  ,draw opacity=1 ]   (170.97,176.84) .. controls (182.27,143.53) and (218.66,136) .. (240,136) ;
\draw [shift={(170,180)}, rotate = 285.52] [fill={rgb, 255:red, 65; green, 117; blue, 5 }  ,fill opacity=1 ][line width=0.08]  [draw opacity=0] (8.93,-4.29) -- (0,0) -- (8.93,4.29) -- cycle    ;
\draw [color={rgb, 255:red, 65; green, 117; blue, 5 }  ,draw opacity=1 ]   (160,170) .. controls (160,140.9) and (198.58,121.21) .. (227.36,120.05) ;
\draw [shift={(230,120)}, rotate = 180] [fill={rgb, 255:red, 65; green, 117; blue, 5 }  ,fill opacity=1 ][line width=0.08]  [draw opacity=0] (8.93,-4.29) -- (0,0) -- (8.93,4.29) -- cycle    ;
\draw [color={rgb, 255:red, 31; green, 76; blue, 124 }  ,draw opacity=1 ]   (180,200) -- (237.15,219.05) ;
\draw [shift={(240,220)}, rotate = 198.43] [fill={rgb, 255:red, 31; green, 76; blue, 124 }  ,fill opacity=1 ][line width=0.08]  [draw opacity=0] (8.93,-4.29) -- (0,0) -- (8.93,4.29) -- cycle    ;
\draw [color={rgb, 255:red, 31; green, 76; blue, 124 }  ,draw opacity=1 ]   (182.85,220.95) -- (240,240) ;
\draw [shift={(180,220)}, rotate = 18.43] [fill={rgb, 255:red, 31; green, 76; blue, 124 }  ,fill opacity=1 ][line width=0.08]  [draw opacity=0] (8.93,-4.29) -- (0,0) -- (8.93,4.29) -- cycle    ;
\draw [color={rgb, 255:red, 208; green, 2; blue, 27 }  ,draw opacity=1 ]   (330,210) -- (387,210) ;
\draw [shift={(390,210)}, rotate = 180] [fill={rgb, 255:red, 208; green, 2; blue, 27 }  ,fill opacity=1 ][line width=0.08]  [draw opacity=0] (8.93,-4.29) -- (0,0) -- (8.93,4.29) -- cycle    ;
\draw [color={rgb, 255:red, 208; green, 2; blue, 27 }  ,draw opacity=1 ]   (333,230) -- (390,230) ;
\draw [shift={(330,230)}, rotate = 0] [fill={rgb, 255:red, 208; green, 2; blue, 27 }  ,fill opacity=1 ][line width=0.08]  [draw opacity=0] (8.93,-4.29) -- (0,0) -- (8.93,4.29) -- cycle    ;
\draw [color={rgb, 255:red, 189; green, 16; blue, 224 }  ,draw opacity=1 ]   (430,170) -- (447.88,187.88) ;
\draw [shift={(450,190)}, rotate = 225] [fill={rgb, 255:red, 189; green, 16; blue, 224 }  ,fill opacity=1 ][line width=0.08]  [draw opacity=0] (8.93,-4.29) -- (0,0) -- (8.93,4.29) -- cycle    ;
\draw [color={rgb, 255:red, 189; green, 16; blue, 224 }  ,draw opacity=1 ]   (412.12,172.12) -- (430,190) ;
\draw [shift={(410,170)}, rotate = 45] [fill={rgb, 255:red, 189; green, 16; blue, 224 }  ,fill opacity=1 ][line width=0.08]  [draw opacity=0] (8.93,-4.29) -- (0,0) -- (8.93,4.29) -- cycle    ;
\draw [color={rgb, 255:red, 189; green, 16; blue, 224 }  ,draw opacity=1 ] [dash pattern={on 4.5pt off 4.5pt}]  (322.34,208.13) -- (370,170) ;
\draw [shift={(320,210)}, rotate = 321.34] [fill={rgb, 255:red, 189; green, 16; blue, 224 }  ,fill opacity=1 ][line width=0.08]  [draw opacity=0] (8.93,-4.29) -- (0,0) -- (8.93,4.29) -- cycle    ;
\draw (215,170) node  {\includegraphics[width=22.5pt,height=15pt]{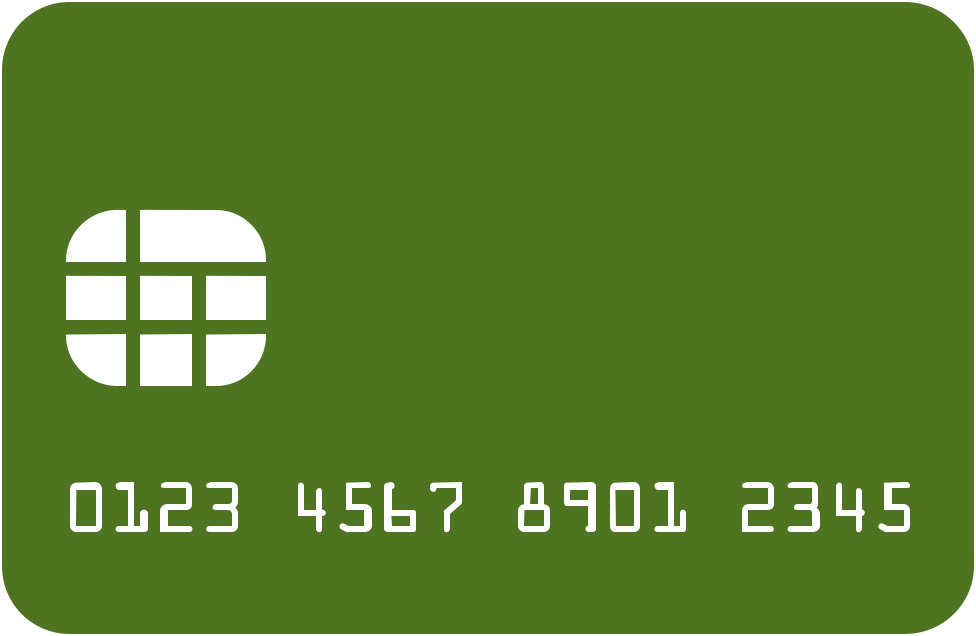}};
\draw (210,245) node  {\includegraphics[width=30pt,height=22.5pt]{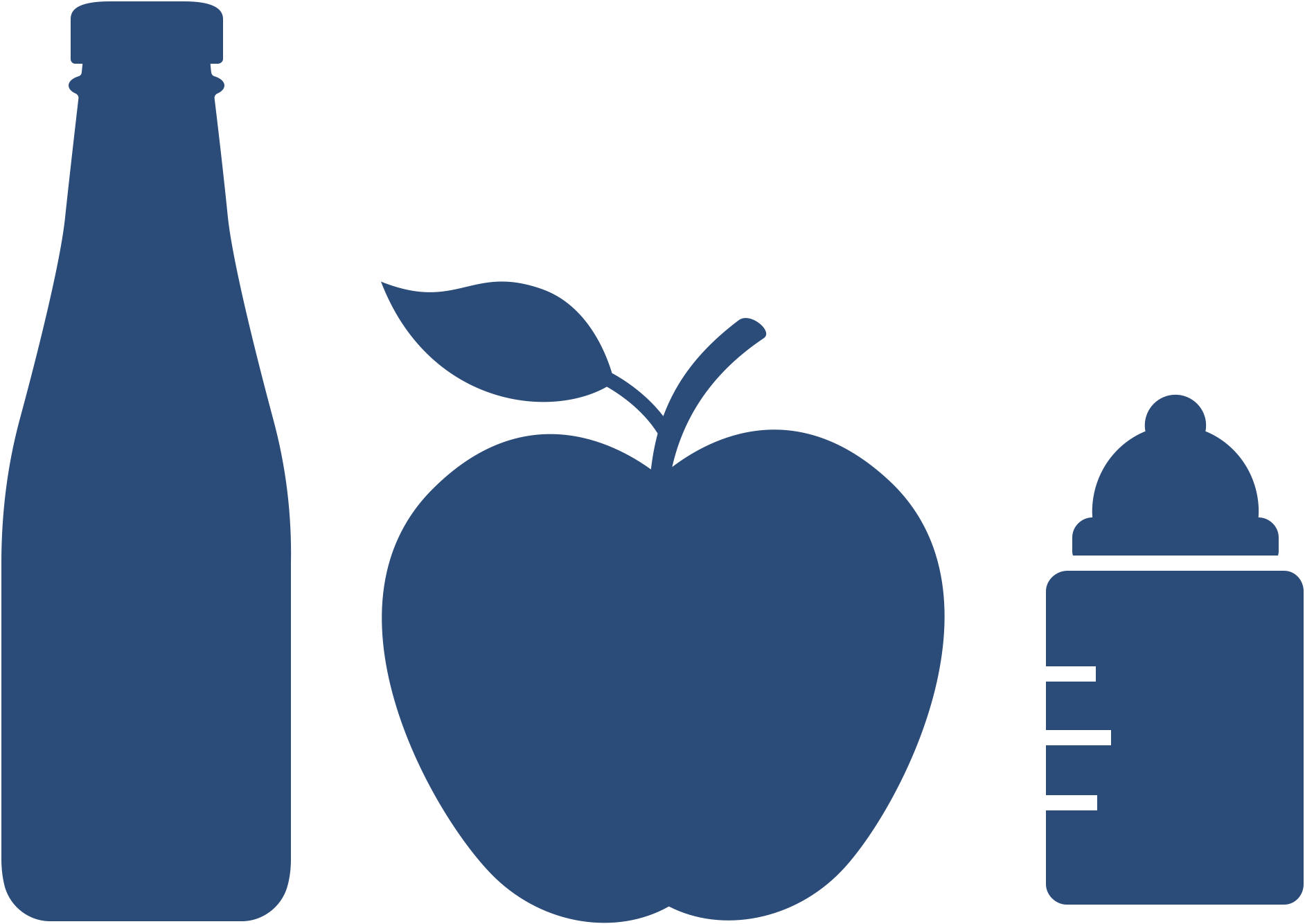}};
\draw (360,245) node  {\includegraphics[width=15pt,height=7.5pt]{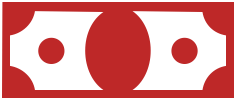}};
\draw (460,170) node  {\includegraphics[width=15pt,height=15pt]{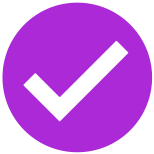}};

\draw (153.5,254) node [anchor=north] [inner sep=0.75pt]   [align=left] {\textbf{User}};
\draw (284,254) node [anchor=north] [inner sep=0.75pt]   [align=left] {\textbf{Vendor}};
\draw (438,254) node [anchor=north] [inner sep=0.75pt]   [align=left] {\textbf{Reclaim Station}};
\draw (282,150) node [anchor=north] [inner sep=0.75pt]   [align=left] {\begin{minipage}[lt]{57.18pt}\setlength\topsep{0pt}
\begin{center}
\textbf{Registration}\\\textbf{Station}
\end{center}
\end{minipage}};
\draw (398,151) node [anchor=north] [inner sep=0.75pt]   [align=left] {\textbf{Auditor}};

\end{tikzpicture}

%% file: parts/3-solution.tex
We now describe our solution.

\subsection{Cryptographic Building Blocks}\label{sec:sol:prelim}
\descr{Notation.}
We write $\code{outp} \leftarrow \code{Alg(inp)}$ to denote an algorithm \code{Alg} run by one party taking as input \code{inp} and returning \code{outp}. We write $\code{\{outp1, outp2\}} \leftarrow \langle\code{Alg1(inp1), Alg2(inp2)}\rangle$ to denote an interactive protocol between two parties. The first party runs \code{Alg1} which takes as input \code{inp1} and returns \code{outp1}. The second party runs \code{Alg2} which takes as input \code{inp2} and returns \code{outp2}. The two algorithms are interactive.
If an algorithm can fail, we write its output \code{outp1/outp2} where \code{outp1} is its output in case of success and \code{outp2} is its output in case of failure. If an algorithm does not have outputs, we denote by $\top$ its success and by $\perp$, its failure.
We write $\langle \code{Alg1} \leftrightarrow \code{Alg2} \rangle$ to denote an interactive protocol where the parties run \code{Alg1} and \code{Alg2} respectively.

We build our construction by combining different cryptographic building blocks that we introduce now. 

\descr{Digital signature scheme.} Let \DS be a digital signature scheme s.t. $\DS\xspace {:=} (\KGen, \Sign, \Verif)$. The key generation algorithm $(\sksig, \pksig) {\larrow} \DS.\KGen(\secuparam)$ takes as input a security parameter $\ell$ and outputs a pair of keys $(\sksig, \pksig)$. The signing algorithm $\sigma {\larrow} \DS.\Sign(\sksig, \mess)$ takes as input a secret key $\sksig$ and a message $\mess \in \{0,1\}^*$ and outputs a signature $\sigma$. The verification algorithm $\top/\perp \larrow \DS.\Verif(\pksig, \mess, \sigma)$ outputs a success $\top$ or an error $\perp$.

\descr{Homomorphic Commitment.} Let \Com be the additively homomorphic Pedersen commitment scheme~\cite{pedersen1991commit} s.t. $\Com\xspace {:=} (\Gen, \Commit)$. The parameters generation algorithm $(\GG, q, g, h) {\larrow} \Com.\Gen(\secuparam)$ takes as input the security parameter $\ell$ and outputs parameters $\params {=} (\GG, q, g, h)$, where $\GG$ is a cyclic group of prime order $q$ generated by $g$ and $h$ is another generator of $\GG$. The commitment algorithm $\cipher {\larrow} \Com.\Commit(\mess, r)$ takes as input a message $\mess$ to commit and a random input $r$, and outputs a commitment $\cipher = g^{\mess} h^r$.

\descr{Message Authentication Code.} Let \MAC be a message authentication code s.t. $\MAC {:=} (\KGen, \Tag, \Verif)$. The key generation algorithm $\code{k} {\larrow} \MAC.\KGen(\secuparam)$ takes as input a security parameter $\ell$ and outputs a key $\code{k}$. The tag algorithm $\tau {\larrow} \MAC.\Tag(\code{k}, \mess)$ takes as input the secret key $\code{k}$ and a message $\mess$ and outputs $\tau$. The verification algorithm $\top/\perp {\larrow} \MAC.\Verif(\code{k}, \mess, \tau)$ takes as input the secret key $\code{k}$, the message $\mess$, and the tag $\tau$ and outputs a success or error symbol $\top/\perp$.\\

\descr{Oblivious RAM.} An oblivious RAM scheme enables a client to access a database stored by an external server while hiding its access patterns~\cite{Goldreich87oram}. Let \ORAM be an Oblivious RAM protocol as $\ORAM\xspace {:=} (\Init, \Read, \Write, \Serve)$. 
We assume the ORAM has integrity protection, i.e., readers can detect if the database has been modified.
The initialization algorithm $(\sk, \db) \larrow \ORAM.\Init(\secuparam, N)$ takes as input the security parameter $\ell$ and a maximal number of stored blocks $N$ and outputs a key $\sk$ and an initialized $\db$. The $\Read$ algorithm is performed by the client and interacts with the $\Serve$ algorithm run by the server $\{\code{data} / \bot, \db\} \!\larrow\! \langle \text{\ORAM.\Read}(\sk, \code{b}), \text{\ORAM.\Serve}(\db)\rangle$. The client takes as inputs the encryption key $\sk$ and the block number $\code{b}$ it is looking for and returns the data stored in this block $\code{data}$ (or $\bot$ if it detects tampering). The server takes as input the database $\db$ and outputs it after the interaction. The $\Write$ algorithm is performed by the client and interacts with the $\Serve$ algorithm run by the server $\{\top/\!\!\perp, \db\} \!\larrow\! \langle \text{\ORAM.\Write}(\sk, \code{b}, \code{data}), \text{\ORAM.\Serve}(\db)\rangle$. The client takes as input the encryption key $\sk$, the block number $\code{b}$ and the data to be written in this block $\code{data}$, and outputs nothing. The server takes as input the database $\db$ and outputs the updated version after the interaction.

Several instantiations of ORAM are possible and we notably use three of them in our construction: a trivial ORAM where clients download and write the entire encrypted database, a tree-based ORAM~\cite{shi2011treeoram, stefanov2018pathoram}, and a recursive ORAM~\cite{shi2011treeoram}. See Appendix~\ref{ap:oram}. The choice of instantiation depends on the volume of data (see \S\ref{sec:eval}).

\begin{figure*}[t]
    \centering
    \resizebox{1.6\columnwidth}{!}{
    \input{images/overview.tikz}
    }
    \captionsetup{font=small} 
    \caption{Overview of the solution. See Section~\ref{sec:sol:syntax} for details.}
    \label{fig:overview}
\end{figure*}

\subsection{Syntax}\label{sec:sol:syntax}
We now present our solution's syntax (see Figure~\ref{fig:overview}).

\descr{Setup.} The trusted party generates the cryptographic parameters shared by all the parties. It also generates an encryption key $\skt$ and initializes the ORAM database. The registration station generates a signing key pair $(\sks, \pks)$. Finally, the trusted party initializes the state stored in each of the smart cards. The setup syntax is:
\begin{itemize}[leftmargin=*]
    \item $(\skt, \pkt, \db) \leftarrow \TrustedSetup(\secuparam, N)$. The trusted party takes as input the security parameter $\ell$ and the maximal number of households $N$. It outputs a key pair $(\skt, \pkt)$ and an initialized database \db.
    \item $(\sks, \pks) \leftarrow \SetupSign(\secuparam)$. The registration station takes as input the security parameter $\ell$ and returns a pair of keys (\sks, \pks).
    \item $\myst \larrow \SetupToken(\pks, \skt, \pkt)$. The smart card takes as input the public key \pks, and the key pair $(\skt, \pkt)$. It returns an initial state \myst.
\end{itemize}

\descr{Registration.} During registration, the user interacts with the registration station which assesses if their household is eligible for aid. If so, the registration station sets a budget $\bud$ and issues a smart card to the user. The syntax of the registration protocol is:
\begingroup
\small
\begin{equation*}
    \{\myst/\!\!\perp, (\Hid, \db)/\db\}\! \,{\larrow}\, \!\langle\Request(\myst), \Allocate(\sks, \bud, \db)\rangle.
\end{equation*}
\endgroup
The smart card runs $\Request$ and takes as input its own state \myst. The registration station runs $\Allocate$ and takes as input the allocated budget \bud, the secret signing key \sks, and the database \db. The smart card outputs its new state \myst in case of success or the error symbol $\perp$ in case of failure. The registration outputs the updated database \db and the household ID \Hid in case of success or only the database \db in case of failure. 

\descr{Transaction.} To perform a transaction, the user unlocks her smart card. The card then interacts with the vendor (or with the shared database via the vendor, see \S\ref{sec:disc}). The transaction syntax is: 

\begingroup
\small
\setlength\abovedisplayskip{-0.1cm}
\setlength\belowdisplayskip{0.1cm}
\[\left\{\begin{matrix} (\success,\price,\varepsilon)/\!\perp,&\\&\hspace{-3.5cm}(\pi,\db)/\!\!\perp\end{matrix}\right\}  \leftarrow \left\langle\begin{matrix}\Spend(\myst, \price),&\\&\hspace{-2cm}\Receive(\varepsilon,\price,\db)\end{matrix}\right\rangle.\]
\endgroup

\noindent The smart card runs $\Spend$ and takes as input its state $\myst$ and the price $\price$. The vendor runs $\Receive$ and takes as input the current reclaim period $\varepsilon$, the price $\price$, and the database \db. The smart card outputs whether the transaction was successful or not. If so, it also outputs the price and the reclaim period provided by the vendor during the protocol. In case of success, the vendor outputs the proof $\pi$ and the updated database $\db$, otherwise it outputs the error symbol $\perp$. The vendor stores the tuple $(\price, \pi)$.

\descr{Reclaim.} After each period, the vendor reclaims the funds for the items it sold. The vendor has to prove that they have sold for the equivalent of a certain amount of money. Using all the proofs of transactions they received from the clients over the period, they compute an aggregated proof for the total sum using $\GenProof$. The reclaim station can then verify this proof using $\code{VerifyReclaimProof}$ and reimburse the vendor with physical cash based on this total amount. The reclaim syntax is:
\begin{itemize}[leftmargin=*]
    \item \begingroup
        \small $(\spendsumsim, \aggproof, \varepsilon) {\leftarrow} \GenProof (\varepsilon, (\spend^{(i)},\allowbreak   \pi^{(i)})_{i})$. 
        \endgroup
        The vendor takes as input the period $\varepsilon$ for which they want to be reclaimed and the stored tuples ($\spend^{(i)}, \pi^{(i)})_{i}$
        for all transactions they want to reclaim. The algorithm outputs the sum of all transactions \spendsumsim, the proof $\aggproof$, and the period $\varepsilon$.
    \item \begingroup
        \small $\top/\!\!\perp \larrow \VerifyProof (\pks, \varepsilon, \spendsumsim, \aggproof).$ 
        \endgroup
        The reclaim station takes as inputs the public signing key \pks, the reclaim period $\varepsilon$, the sum of transactions for this period \spendsumsim, and the proof $\aggproof$. The algorithm outputs $\top$ in case of success and $\perp$ otherwise. 
\end{itemize}

\descr{Auditing.} The auditor checks the consistency of the \ICRC's spending. When an audit takes place, the \ICRC\ provides the auditor with the aggregated proof computed during the reclaim phase who will check it. The auditing syntax is the same as the verification of the reclaim proof.
As the reclaim process is not interactive, if the auditor accepts the proof, it will accept the allocated budget at the same time.

\subsection{Solution Details}\label{sec:sol:soldetails}
We now detail how our solution operates.

\descr{Setup.} 
We specify the instantiation of the three setup algorithms:
\begin{itemize}[leftmargin=*]
    \item $(\skt, \pkt, \db) {\larrow} \TrustedSetup(\secuparam, N)$. The trusted party sets up the Pedersen commitment scheme by calling $\params {\leftarrow} \Com.\Gen(\secuparam)$ where $\params=(\GG, q, g, h)$. It publishes $\params$. It then runs the ORAM initialization that returns $(\sk, \db) \gets \ORAM.\Init(\secuparam, N)$ where $\sk$ is the ORAM key and $\db$ is an initialized encrypted database. It also generates a key $\prfk$ for a pseudorandom function (PRF).  
    Finally, the trusted party sets $\skt=(\sk, \prfk)$, $\pkt=\perp$ and returns $(\skt, \pkt, \db)$.\footnote{We include the public key $\pkt$ as output of $\TrustedSetup$ for consistency with Wang \etal~\cite{wang2023digital}.}
    \item $(\sks, \pks) {\larrow} \SetupSign(\secuparam)$. The registration station takes as input \secuparam and runs (\sks, \pks) \larrow \DS.\KGen(\secuparam) to obtain a signing key pair. 
    \item $\myst {\larrow} \SetupToken(\pks, \skt, \pkt)$. The trusted party sets up the smart card by securely storing the public key of the registration station \pks, the keys $\skt = (\sk, \prfk)$ and $\pkt$. Let $\skt' = \sk$, and return the state $\myst = (\pks, \skt', \pkt, \prfk)$.
\end{itemize}

\descr{Registration.} 
To register a household, the registration station (RS) issues a smart card to one user and runs the protocol depicted in Figure~\ref{fig:registration}.\footnote{W.l.o.g. we consider the registration of the first member of a household. See Appendix~\ref{append:whisper} and prior work~\cite{wang2023digital} for details.}

\begin{figure}
    \centering
    \begin{pchstack}[space=1em, center, boxed, space=0em]
    \procedure[linenumbering, width=0.1\textwidth, codesize=\scriptsize, jot=-0.3mm,xshift=-0.5mm]{\scriptsize Smart~Card:\,$\Request(\myst)$}{
        \label{my:line:parse}(\pks, \skt,  \pkt, \prfk) {\leftarrow} \myst \\
        \label{my:line:generate_values}\\
        \label{my:line:send_id}\code{Receives } \Hid\\
        \label{my:line:verify_sks_start} \code{Receives } \sks \\
        \sigma {\larrow}\code{DS.Sign(\sks, } r), \code{r} {\in}_\mathcal{R} \ZZ_q\\
        \pcif \neg\code{ DS.Verif}(\pks, r, \sigma) \pcthen \\
        \label{my:line:verify_sks_end}\pcind \code{ abort and go to }\ref{my:line:abort}\\
        \label{my:line:send_bud} \code{Receives }\bud\\
        \label{my:line:write_bud}  \code{ORAM.Write}(\skt, \Hid, \bud \!\parallel\! 0)\\
        \label{my:line:state}\myst \larrow (\sks, \pks,\skt, \pkt, \prfk, \Hid)\\
        \label{my:line:output} \pcreturn \myst\\
        \label{my:line:abort} \pcreturn \perp
    }
    \procedure[width=0.20\textwidth, codesize=\scriptsize,jot=-0.3mm]{\scriptsize RS: $\Allocate(\sks, \bud, \db)$}{
        \\
        \code{Picks next available } \Hid \\
        \code{Sends } \Hid\\
        \code{Sends } \sks \pcskipln\\
        \\
        \\
        \\
        \code{Sends } \bud\\
        \code{ORAM.Serve}(\db)\\
        \\
        \pcreturn (\Hid, \db)\\
        \pcreturn \db
    }
    \end{pchstack}
    \caption{Registration protocol}
    \label{fig:registration}
\end{figure}

\begin{enumerate}[leftmargin=*]
    \item The smart card takes as input its state \st\xspace and the RS takes as input the allocated budget \bud, the secret signing key \sks and the database \db.
    \item The smart card parses \myst as $(\pks, \skt, \pkt, \prfk)$ (line 1).
    \item The RS picks the next available household ID \Hid, and sends it to the smart card (lines 2--3).
    \item The RS sends \sks to the smart card which verifies it with respect to \pks (lines 5--7).
    \item  The RS sends the allocated budget \bud to the smart card (line 8)
    \item The smart card writes its budget and a fresh counter $\ctr=0$ in the \db using \ORAM.\Write($\cdot$) (line 9).
    \item The smart card outputs $\myst = (\sks, \pks, \skt,\allowbreak \pkt, \prfk, \Hid)$ or $\perp$ in case of failure. The RS outputs \Hid and \db or only \db in case of failure. (Lines 11--12.)
\end{enumerate}

\descr{Transaction.} 
When executing a transaction, the user first authenticates to the card and unlocks it. As in Wang~\etal~\cite{wang2023digital}, users can authenticate using for example a PIN or biometrics. The unlocked smart card checks if the remaining balance of the household is higher than the agreed price for the transaction. If so, the smart card sends the vendor a \textit{proof of transaction} valid for the current reclaim period and updates the balance. If the proof is valid, the transaction is successful. 
We describe how the smart card interacts with the vendor to proceed to a transaction as presented in Figure~\ref{fig:transaction}. 

\begin{figure}
    \centering   
    \begin{pchstack}[boxed, space=0em]
        \procedure[linenumbering, width=0.1\textwidth, codesize=\scriptsize,jot=-0.3mm,xshift=-0.5mm]
            {\scriptsize Smart~Card:\,$\Spend(\myst, \price)$}{%
            \label{my:line:unpack}(\sks, \pks, \skt,  \pkt, \prfk, \Hid) {\leftarrow} \myst \\ 
            \label{my:line:exchange_values} \code{Receives } \price \code{ and } \varepsilon\\
            \label{my:line:retrieve}\bal \!\parallel\! \ctr {\leftarrow} \code{ORAM.Read}(\skt,  \Hid)\\
            \label{my:line:check_money}\code{b} \leftarrow \price \leq \bal\\
            \pcif  \neg\code{ b} \pcthen \code{abort and go to } \ref{my:line:failure}\\
            \label{my:line:update_bal}\bal \leftarrow \bal - \price\\
            \label{my:line:update_ctr}\ctr \leftarrow \ctr + 1 \\
            \label{my:line:commit}\code{Com} {=} \code{Com.Commit}(\price,  \code{r}), \code{r} {\in}_\mathcal{R} \ZZ_q\\
            \label{my:line:hashctr} \tau = \code{PRF}_{\prfk}(\Hid \parallel \ctr)\\
            \label{my:line:sign}\sigma = \code{DS.Sign}   (\sks, \tau \!\parallel\! \varepsilon \!\parallel\! \code{Com})\\
            \label{my:line:proof}\pi = (\sigma, \tau, \code{Com}, \code{r})\\
            \label{my:line:update_db}\code{ORAM.Write}(\skt, \Hid, \bal \!\parallel\! \ctr) \\
            \label{my:line:parse_proof}\code{Sends } \pi \\
            \label{my:line:check_proof} \pcskipln \\ 
            \\
            \\
            \label{my:line:success}\pcreturn (\success, \price, \varepsilon)\\
            \label{my:line:failure}\pcreturn \perp
            }
    
        \procedure[width=0.21\textwidth, codesize=\scriptsize,jot=-0.3mm]
            {\scriptsize Vendor:\,$\Receive(\varepsilon, \price, \db)$}{%
            \\
            \code{Sends } \price \code{ and } \varepsilon\\
            \db \leftarrow \code{ORAM.Serve}(\db)\\
            \\
            \\
            \\
            \\
            \\
            \\
            \\
            \\
            \code{ORAM.Serve}(\db) \pcskipln\\
            \code{Parses } \pi = (\sigma, \tau, \code{Com}, \code{r})\\
            \code{b} {\leftarrow} \code{DS.Verif}  (\pks, \sigma, \tau \!\parallel\! \varepsilon \!\parallel\!  \code{Com}) \pcskipln\\
            \pcind \wedge \code{Com} {=} \code{Com.Commit} (\price, \code{r})\\
            \pcif \neg\code{ b} \pcthen \code{abort and go to } \ref{my:line:failure}\\
            \pcreturn (\pi, \db)\\
            \pcreturn \perp
            }
        \end{pchstack}
    \caption{Transaction protocol}
    \label{fig:transaction}
\end{figure}

\begin{enumerate}[leftmargin=*]
    \item The smart card takes as input its state \myst and the vendor takes as input ($\varepsilon$, \price, \db).
    \item The smart card unpack its state \myst in order to access its keys and household identifier \Hid (line 1).
    \item The vendor sends \price and current reclaim period $\varepsilon$ to the smart card (line 2).
    \item The smart card and the vendor interact so that the smart card retrieves its current \bal and the counter \ctr from the \db privately using $\ORAM.\Read$ (line 3).
    \item The smart card checks if it has enough budget left and if so updates its balance (lines 4--6). The smart card updates the counter of transactions \ctr (line 7), creates a commitment \Com (line 8), creates a transaction tag $\tau$ using $\code{PRF}_{\prfk}$ (line 9), signs the transaction tag, period, and commitment (line 10), and produces the corresponding proof $\pi$ (line 11). The smart card sends the proof to the vendor and the latter parses it (line 13).
    \item The vendor verifies the proof (line 14).
    \item The smart card outputs $((\success, \price, \varepsilon)/\perp)$, and the vendor outputs (($\pi$, \db)/$\perp$).
\end{enumerate}

\descr{Reclaim.} 
The vendor creates from all the transaction proofs a reclaim proof that it uses to be reclaimed. 

\begin{itemize}[leftmargin=*]
\setlength\itemsep{0.2em}
    \item \textbf{Preparation phase}. The vendor creates the reclaim proof by running $\GenProof$. It takes as input the period $\varepsilon$, and a set of transactions and proofs $(\spend^{(i)}, \pi^{(i)})_{i}$. The vendor unpacks the proof $\pi^{(i)} {=} (\sigma^{(i)}\!, \tau^{(i)}, \code{Com}^{(i)}\!, \code{r}^{(i)})$ and computes the sum of amounts and the sum of random values
    \begin{center}
        $\spendsumsim = \sum_i \spend^{(i)}, \quad \code{ r}_\code{sum} =\sum_i \code{r}^{(i)}.$
    \end{center}
    Then, the vendor creates the proof
    \begin{center}
        $\aggproof = (\code{r}_\code{sum}, (\sigma^{(i)}, \tau^{(i)}, \code{Com}^{(i)})_i),$
    \end{center}
    and returns it with the sum of spent amount \spendsumsim and the reclaim period $\varepsilon$. 

    \item \textbf{Verification phase}. The reclaim station checks the proof. 
    
    The \code{VerifyReclaimProof} algorithm takes as input the signing public key \pks, the reclaim period $\varepsilon$, the total spent amount \spendsumsim and the reclaim proof $\aggproof$. The reclaim station checks:
    \begin{itemize}[leftmargin=*]
        \item the validity of all the signatures in $\aggproof$
        
        \begin{center}
            $\code{V} {=} \code{DS.Verif}(\pks, \sigma^{(i)}, \tau^{(i)} \parallel \varepsilon \parallel \code{Com}^{(i)})$,      
        \end{center}
        
        \item the sum of the spent amount and the sum of random values match the commitment
        
        \begin{center}
            $\prod_i \code{Com}^{(i)} {=} \code{Com.Commit}(\spendsumsim, \code{r}_\code{sum}),$
        \end{center}

        \item that each transaction tag $\tau$ is unique.
    \end{itemize}
    If all checks pass it outputs $\top$ and $\perp$ otherwise.
\end{itemize}

\descr{Auditing.} 
The auditor verifies the aggregated proof the same way the reclaim station did. It reproduces the verification phase from the reclaim phase.

%% file: images/overview.tikz
\tikzset{every picture/.style={line width=0.75pt}} 

\begin{tikzpicture}[x=0.75pt,y=0.75pt,yscale=-1,xscale=1]

\draw    (100,48) -- (100,320) ;
\draw    (221,48) -- (221,320) ;
\draw    (341,48) -- (341,320) ;
\draw    (580,48) -- (580,320) ;
\draw    (700,48) -- (700,320) ;
\draw    (460,48) -- (460,320) ;

\draw  [color={rgb, 255:red, 245; green, 166; blue, 35 }  ,draw opacity=1 ] (10,50) -- (780,50) -- (780,118) -- (10,118) -- cycle ;
\draw  [color={rgb, 255:red, 65; green, 117; blue, 5 }  ,draw opacity=1 ] (10,120) -- (780,120) -- (780,178) -- (10,178) -- cycle ;
\draw  [color={rgb, 255:red, 31; green, 76; blue, 124 }  ,draw opacity=1 ] (10,180) -- (780,180) -- (780,238) -- (10,238) -- cycle ;
\draw  [color={rgb, 255:red, 208; green, 2; blue, 27 }  ,draw opacity=1 ] (10,240) -- (780,240) -- (780,288) -- (10,288) -- cycle ;
\draw  [color={rgb, 255:red, 189; green, 16; blue, 224 }  ,draw opacity=1 ] (10,290) -- (780,290) -- (780,318) -- (10,318) -- cycle ;

\draw   (210,130) -- (350,130) -- (350,170) -- (210,170) -- cycle ;
\draw   (330,190) -- (470,190) -- (470,230) -- (330,230) -- cycle ;
\draw  [draw opacity=0][fill={rgb, 255:red, 255; green, 255; blue, 255 }  ,fill opacity=1 ] (70,55) -- (140,55) -- (140,90) -- (70,90) -- cycle ;

\draw (80,38) node [inner sep=0.75pt]  [font=\normalsize] [align=left] {\textbf{Trusted Party}};
\draw (221,38) node [inner sep=0.75pt]  [font=\normalsize] [align=left] {\textbf{Registration Station}};
\draw (345,38) node [inner sep=0.75pt]  [font=\normalsize] [align=left] {\textbf{Smart Card}};
\draw (460,38) node [inner sep=0.75pt]  [font=\normalsize] [align=left] {\textbf{Vendor}};
\draw (580,38) node [inner sep=0.75pt]  [font=\normalsize] [align=left] {\textbf{Reclaim Station}};
\draw (700,38) node [inner sep=0.75pt]  [font=\normalsize] [align=left] {\textbf{Auditor}};

\draw (11,53) node [anchor=north west][inner sep=0.75pt]  [font=\normalsize] [align=left] {\textcolor[rgb]{0.96,0.65,0.14}{\textbf{Setup}}};
\draw (11,123) node [anchor=north west][inner sep=0.75pt]  [font=\normalsize] [align=left] {\textcolor[rgb]{0.25,0.46,0.02}{\textbf{Registration}}};
\draw (11,183) node [anchor=north west][inner sep=0.75pt]  [font=\normalsize] [align=left] {\textcolor[rgb]{0.12,0.3,0.49}{\textbf{Transaction}}};
\draw (11,243) node [anchor=north west][inner sep=0.75pt]  [font=\normalsize] [align=left] {\textcolor[rgb]{0.82,0.01,0.11}{\textbf{Reclaim}}};
\draw (11,293) node [anchor=north west][inner sep=0.75pt]  [font=\normalsize] [align=left] {\textcolor[rgb]{0.74,0.06,0.88}{\textbf{Audit}}};

\draw (13,84) node [anchor=south west] [inner sep=0.75pt]    {$(\skt, \pkt, \db ){\larrow} \TrustedSetup(\secuparam)$};
\draw  [draw opacity=0][fill={rgb, 255:red, 255; green, 255; blue, 255 }  ,fill opacity=1 ]  (208,85) -- (231,85) -- (231,102) -- (208,102) -- cycle  ;
\draw (219.5,102) node [anchor=south] [inner sep=0.75pt]    {$(\sks, \pks) {\larrow} \SetupSign(\secuparam)$};
\draw  [draw opacity=0][fill={rgb, 255:red, 255; green, 255; blue, 255 }  ,fill opacity=1 ]  (328,60) -- (352,60) -- (352,82) -- (328,82) -- cycle  ;
\draw (340,82) node [anchor=south] [inner sep=0.75pt]    {$\myst {\larrow} \SetupToken(\pks, \skt, \pkt)$};
\draw (209,140) node [anchor=south east] [inner sep=0.75pt]    {$(\sks, \bud, \db)$};
\draw (208,166.6) node [anchor=south east] [inner sep=0.75pt]    {$(\Hid, \db)/\db$};
\draw (351,137.6) node [anchor=south west] [inner sep=0.75pt]    {$\myst$};
\draw (352,166.6) node [anchor=south west] [inner sep=0.75pt]    {$\myst /\!\!\perp$};
\draw (280,150) node    {$ \langle \Request {\leftrightarrow} \Allocate\rangle$};
\draw (329,199) node [anchor=south east] [inner sep=0.75pt]    {$(\myst, \price)$};
\draw  [draw opacity=0][fill={rgb, 255:red, 255; green, 255; blue, 255 }  ,fill opacity=1 ]  (207,212) -- (232,212) -- (232,227) -- (207,227) -- cycle  ;
\draw   [color={rgb, 255:red, 211; green, 211; blue, 211 }  ,draw opacity=1 ] (221,212) -- (221,227) ;
\draw (328,226.6) node [anchor=south east] [inner sep=0.75pt]    {$(\success,\price,\varepsilon) / \!\perp$};
\draw (471,199) node [anchor=south west] [inner sep=0.75pt]    {$(\varepsilon,\price,\db)$};
\draw (472.21,226.84) node [anchor=south west] [inner sep=0.75pt]  [rotate=-358.45]  {$(\pi,\db)/\!\perp$};
\draw (400,210) node    {$\langle \Spend {\leftrightarrow} \Receive \rangle$};

\draw  [draw opacity=0][fill={rgb, 255:red, 255; green, 255; blue, 255 }  ,fill opacity=1 ]  (300,243) -- (600,243) -- (600,260) -- (300,260) -- cycle  ;
\draw   [color={rgb, 255:red, 211; green, 211; blue, 211 }  ,draw opacity=1 ] (341,243) -- (341,260) ;
\draw   [color={rgb, 255:red, 211; green, 211; blue, 211 }  ,draw opacity=1 ] (580,243) -- (580,260) ;
\draw (459.5,260) node [anchor=south] [inner sep=0.75pt]    {$(\spendsumsim, \aggproof, \varepsilon) {\leftarrow} \GenProof (\varepsilon, (\spend^{(i)}\!\!\!, \pi^{(i)})_{i})$};

\draw  [draw opacity=0][fill={rgb, 255:red, 255; green, 255; blue, 255 }  ,fill opacity=1 ]  (368,264) -- (705,264) -- (705,282) -- (368,282) -- cycle  ;
\draw   [color={rgb, 255:red, 211; green, 211; blue, 211 }  ,draw opacity=1 ] (460,264) -- (460,282) ;
\draw   [color={rgb, 255:red, 211; green, 211; blue, 211 }  ,draw opacity=1 ] (700,264) -- (700,282) ;
\draw (579.5,280) node [anchor=south] [inner sep=0.75pt]    {$\top/\!\!\perp\!\! {\larrow} \VerifyProof (\pks, \varepsilon, \spendsumsim, \aggproof\!)$};

\draw  [draw opacity=0][fill={rgb, 255:red, 255; green, 255; blue, 255 }  ,fill opacity=1 ]  (699,294) -- (705,294) -- (705,312) -- (699,312) -- cycle  ;
\draw   [color={rgb, 255:red, 211; green, 211; blue, 211 }  ,draw opacity=1 ] (580,294) -- (580,312) ;
\draw (768,312) node [anchor=south east] [inner sep=0.75pt]    {$\top/\!\!\perp \larrow \VerifyProof (\pks, \varepsilon, \spendsumsim, \aggproof)$};

\end{tikzpicture}

%% file: parts/4-secuAndPrivacy.tex
We now analyze the security and privacy of our construction. In the following, we consider the smart card token to be a Secure Element (SE) providing both tamper-resistant hardware and secure storage~\cite{ShepherdAGLMASC16}, i.e., smart cards honestly execute the protocols and secrets cannot be extracted. 

\subsection{Oracles}

We design oracles that enable the adversary to interact with honest and malicious parties and keep track of global variables that are accessed in the different games.

\subsubsection{Registration oracles}

We create registration oracles in Algorithm~\ref{alg:reg-oracles}. They enable the adversary to register a fixed number of smart cards per household under different threat models: honest parties, a malicious user, and a malicious registration station. 

\begin{itemize}[leftmargin=*]
    \item \oraclehregargs: this oracle enables the adversary to register a chosen number of smart cards \tnb for an honest household by interacting with an honest registration station. The oracle keeps track of the allocated budget per household in a set of all budgets \Bud. For each registered smart card, it also creates smart card ID \tid by incrementing a counter $\code{counter}$ (initialized to~0). These identifiers allow a malicious vendor to choose which smart card they want to interact with later on. Finally, the oracle tracks the state of the smart cards in \SMap\ and the allocated smart card IDs \tid in the set \myTH.
    \item \oraclemuregargs: this oracle enables the adversary to (maliciously) register a new household interacting with an honest registration station.  The oracle tracks these malicious households by adding their \Hid to the set of malicious IDs \myM and storing the allocated budget \bud in the set of budgets \Bud.
    \item \oraclemsregargs: this oracle enables the adversary to register an honest household to an honest-but-curious registration station. The adversary can choose the smart card IDs \code{ChosenIds} as long as they have not been used before. The oracle tracks the state of the smart cards in $\SMap$ and returns the full transcript to the adversary.
\end{itemize}

\begin{algorithm}[t]
    \caption{Registration Oracles}\label{alg:reg-oracles}
    \begin{algorithmic}[1]
      \footnotesize
      \Function{\oraclehregargs}{\bud, \sks, \tnb}\Comment{Honest parties} 
        \State $\myst \larrow \SetupToken(\pks, \skt, \pkt)$
        \State ${\{\myst/\!\!\perp, (\Hid, \db)/\db\}} {\larrow} \langle \Request(\myst), \!\Allocate(\sks, \bud, \db)\rangle$
        \State $\Bud[\Hid] \larrow \bud$
        \State $\pcfor i \in \{1, \ldots, \tnb\}$
        \State\hskip\algorithmicindent $\tid \leftarrow \counter + 1$
        \State\hskip\algorithmicindent $\SMap[\tid] \larrow \myst$
        \State\hskip\algorithmicindent $\code{NewIds} \leftarrow \code{NewIds} \cup \{\tid\}$
        \State$\myTH \larrow \myTH \cup \{\code{NewIds}\}$ 
        \State $\pcreturn \code{NewIds}$
      \EndFunction
      \vspace{1mm}

      \Function{\oraclemuregargs}{\bud, \sks}\Comment{Malicious user, honest RS}
        \State $\{ - , (\Hid, \db)/\db\} \larrow \langle \adv, \Allocate(\sks, \bud,\db)\rangle$
        \State $\myM \larrow \myM \cup \{\Hid\}$ 
        \State $\Bud[\Hid] \leftarrow \bud$
      \EndFunction
      \vspace{1mm}

      \Function{\oraclemsregargs}{\code{ChosenIds}, \bud}\Comment{Hon. user, cur. RS}
        \State $\myst \larrow \SetupToken(\pks, \skt, \pkt)$
        \State ${\{\myst/\!\!\perp, (\Hid, \db)/\db\}} {\larrow} \langle \Request(\myst), \!\Allocate(\sks, \bud, \db)\rangle$
        \State Let $\code{trs}$ be the transcript of the previous interaction.
        \State $\pcfor \tid \in \code{ChosenIds}$
        \State\hskip\algorithmicindent $\pcif \tid \notin \myTH$
        \State\hskip\algorithmicindent\hskip\algorithmicindent $\SMap[\tid] \larrow \myst$
        \State\hskip\algorithmicindent\hskip\algorithmicindent $\code{NewIds} \leftarrow \code{NewIds} \cup \{\tid\}$
        \State $\pcreturn (\code{NewIds}, \code{trs})$ 
      \EndFunction
    \end{algorithmic}
\end{algorithm}

\subsubsection{Transaction oracles}

Similarly, we create transaction oracles in Algorithm~\ref{alg:spend-oracles} that enable the adversary to execute a transaction under different threat models. 

\begin{itemize}[leftmargin=*]
    \item \oraclespendargs: This oracle simulates a normal interaction between an honest user (with an honest smart card) and an honest vendor. The user uses a smart card identified by \tid. 
    The oracle takes as input transaction-related information: the reclaim period $\varepsilon$ and the price $\price$ of the transaction. The oracle runs the spend transaction. 
    If the interaction is successful, the oracle adds the transaction to the record of received transactions $\inlog[\varepsilon]$ for the specific reclaim period $\varepsilon$, as well as to the spent transactions of honest users $\spenttrans[\varepsilon]$.
    \item \oraclespendmuargs: This oracle simulates the same situation as \oraclespendargs, except that the smart card is now played by the adversary. The oracle records the transaction output in $\inlog[\varepsilon]$.
    \item \oraclespendmvargs: This oracle simulates the same situation as \oraclespendargs except that the vendor is now played by the adversary. The malicious vendor can choose with which smart card they want to interact. The game records the transaction output in $\spenttrans[\varepsilon]$.
\end{itemize}

\begin{algorithm}[t]
    \caption{Spend Oracles}\label{alg:spend-oracles}
    \begin{algorithmic}[1]
      \footnotesize 
      \Function{\oraclespendargs}{$\varepsilon, \tid, \price$}\Comment{Honest parties} 
        \State $\myst \leftarrow \SMap[\tid] \code{ or abort if does not exist}$
        \State $\{(\success, \price, \varepsilon) / \perp, (\pi, \db) / \perp\}\leftarrow\langle \Spend(\myst, \price), \Receive(\varepsilon, \price, \db)\rangle$
        \State $\pcif \success \pcthen$
        \State \hskip\algorithmicindent $\inlog[\varepsilon] \leftarrow \inlog[\varepsilon] \cup \{\price\}$
        \State \hskip\algorithmicindent $\spenttrans[\varepsilon] \leftarrow \spenttrans[\varepsilon] \cup \{\price\}$
      \EndFunction
      \vspace{1mm}

      \Function{\oraclespendmuargs}{$\varepsilon, \amountspent$}\Comment{Mal. user}
        \State $\{ -, (\pi, \db)/\perp\}\leftarrow \langle\adv, \Receive(\varepsilon, \amountspent, \db)\rangle$
        \State $\pcif \code{vendor outputs } \neq \perp \pcthen$
        \State\hskip\algorithmicindent$\inlog[\varepsilon]\leftarrow \inlog[\varepsilon] \cup \{\amountspent\}$
      \EndFunction
      \vspace{1mm}

      \Function{\oraclespendmvargs}{$\varepsilon, \tid, \amountrec$}\Comment{Mal. vendor}
        \State $\myst \leftarrow \SMap[\tid] \code{ or abort if does not exist}$
        \State $\{(\success, \amountrec', \varepsilon')/\perp, - \} \leftarrow \langle\Spend(\myst, \price), \adv\rangle$
        \State $\pcif (\success \wedge \amountrec' = \amountrec \wedge \varepsilon' = \varepsilon) \pcthen$
        \State \hskip\algorithmicindent $\spenttrans[\varepsilon] \leftarrow \spenttrans[\varepsilon] \cup \{\amountrec\}$ 
      \EndFunction
      \vspace{1mm}
    \end{algorithmic}
\end{algorithm}

\subsection{Security of spending}

\begin{algorithm}[t]
    \caption{Overspending Experiment}\label{alg:overspend-exp}
    \begin{algorithmic}[1]
      \footnotesize
      \Function{\expsec}{\secuparam}\Comment{Malicious client, honest vendor} 
        \State $(\skt, \pkt, \db) \leftarrow \TrustedSetup(\secuparam, N)$
        \State $(\sks, \pks) \leftarrow \SetupSign(\secuparam)$
        \State $\varepsilon\textsuperscript{*}{ \leftarrow} \adv^{\substack{\oraclehreg, \oraclemureg,  \oraclespend,\\ \oraclespendmu, \oraclespendmv}}(\pks)$
        \State $\spendseen {\leftarrow}\! \sum {\inlog}[\varepsilon\textsuperscript{*}]$
        \State $\spendmax \leftarrow \sum \spenttrans[\varepsilon \textsuperscript{*}] + \sum_{\Hid \in \myM} \Bud[\Hid]$
        \State $\pcreturn \spendseen > \spendmax$
      \EndFunction
      \vspace{1mm}
    \end{algorithmic}
\end{algorithm}

The security of spending means that legitimate users cannot spend more than their allocated budget~(\reqlinky{overspending}) and illegitimate users cannot make any spending. We model this using the security experiment in Algorithm~\ref{alg:overspend-exp}. The adversary \adv\xspace can interact with the oracles $\oraclehregargs(\cdot)$ and \oraclespend to interact with an honest registration station and an honest vendor as an honest legitimate user. The adversary has access to the $\oraclemuregargs(\cdot)$ to register a malicious household and create several smart cards to be used in a chosen order. The adversary also has access to the \oraclespendmu to interact with an honest vendor as a malicious user trying to make a purchase. This oracle keeps track of all the transactions that succeeded, i.e., the amount that is spent on each successful transaction. Finally, the adversary also has access to the \oraclespendmv oracle which can simulate a machine-in-the-middle attack when used with \oraclespendmu.
After interacting with the oracles, the adversary outputs a target reclaim period $\varepsilon^*$ during which the transactions were performed. First, the challenger computes the total amount of money spent in successful calls to \oraclespend and \oraclespendmu during the chosen period. Then, the challenger computes the maximum spendable amount as all amounts spent by honest users (recorded in $\spenttrans[\varepsilon^{*}]$) and the sum of all budgets allocated to malicious households. If the amount spent is greater than the spendable sum, the adversary wins. As several malicious users could collude, we compare the sum of all transactions with the sum of all budgets.

We note that the adversary could perform a machine-in-the-middle attack during registration. Yet, in this situation, both the user and the registration station are honest and communicate via an authenticated channel. Hence, we do not model this situation.

Because the user needs to authenticate to the smart cards, illegitimate users cannot unlock the cards. As mentioned in \S\ref{sec:sol:soldetails}, actual deployments have a PIN or biometric mechanism in place to bind cards to their holder. We therefore do not model card theft in our games.

\begin{definition} The system is secure against overspending if the following probability is negligible:
\begin{center}
    $\code{Succ}^{\text{SEC}}(\adv) = \code{Pr}[\expsec(\secuparam) = 1]$
\end{center}
\end{definition}
\begin{theorem}
The system is secure against overspending providing that the smart card is a SE and the channel with smart cards is authenticated.
\end{theorem}

\begin{proof}[Proof Sketch] We treat legitimate and illegitimate users separately.

\textit{Legitimate.} 
As we model the smart card as a SE, the adversary does not have access to \oraclemuregargs and \oraclespendmuargs, but it can use honest cards in transactions (using \oraclespend). Honest cards, however, always check the latest balance in the database. So the only way the adversary could win is by making a smart card believe it has more budget left in the balance than in reality. However, to achieve this, the adversary would need to modify the balance in the database, but it cannot as the vendors and registration station are trusted. 

\textit{Illegitimate.} 
An illegitimate user is not provided with a smart card and cannot have access to the registration. Thus, the adversary does not have access to the registration oracles. Yet, it has access to the spending oracle. We discuss two attacks from the illegitimate user:
    \begin{itemize}
        \item The illegitimate user builds their smart card. To simulate this attack, the adversary calls the \oraclespendmu oracle. Yet, as we model the card as a SE, it does not know the secret key \skt and cannot create a valid proof $\pi$ for the vendor. The transaction fails.
        \item The illegitimate user performs a machine-in-the-middle attack. To simulate this attack, the adversary uses the \oraclespendmv and \oraclespendmu to relay information. However, it cannot modify it as the honest smart card and honest vendor set up an authenticated channel that prevents tampering. \qedhere
    \end{itemize}
\end{proof}

\subsection{Security of reclaim}

\begin{algorithm}[t]
    \caption{Over-reclaiming experiments}\label{alg:overreclaim-exp}
    \begin{algorithmic}[1]
      \footnotesize
      \Function{\expred}{\secuparam}\Comment{Mal. vendor, mal. users}
        \State $(\skt, \pkt) \leftarrow \TrustedSetup(\secuparam)$
        \State $(\sks, \pks) \leftarrow \SetupSign(\secuparam)$
        \State $\varepsilon\textsuperscript{*}, \spendsum, \pi^{*}_{\code{red}}\leftarrow\adv ^{\substack{\oraclehreg, \oraclemureg,\\ \oraclespendmv}}()$
        \State $\valid \leftarrow \VerifyRedProof (\pks, \varepsilon^{*}, \spendsum,\pi^{*}_{\code{red}})$
        \State $\spendmax \leftarrow \spenttrans[\varepsilon \textsuperscript{*}] + \sum_{\Hid \in \myM} \Bud[\Hid]$
        \State $\pcreturn \valid \wedge (\spendsum > \spendmax)$
      \EndFunction
    \end{algorithmic}
\end{algorithm}

A vendor should not be able to reclaim more than what was spent in their shop~(\reqlinky{over-reclaim}). We model this using the reclaim experiment in Algorithm~\ref{alg:overreclaim-exp}. Adversary \adv\xspace plays the role of a malicious vendor who might collude with malicious users as well.
The adversary has access to $\oraclehregargs(\cdot)$, and $\oraclemuregargs(\cdot)$ to register honest and malicious users. This enables them to create different smart cards and use them in a chosen order. Finally, the adversary has access to $\oraclespendmv$ to interact with users as a malicious vendor. After interacting with the oracles, the adversary outputs a chosen reclaim period $\varepsilon^*$, the amount of money it claims to have received $\spendsum$ and the corresponding proof $\pi^{*}_{\code{red}}$. First, the challenger verifies that the proof is valid for the amount $\spendsum$. Then, the challenger computes the maximum spendable amount as the sum of all allocated budgets. The adversary wins if the proof is valid and the reclaimed amount is greater than the allocated budget.
\begin{definition}\label{def:overreclaim}
The system is secure against over-reclaiming if the following probability is negligible:
\begin{center}
    $\code{Succ}^{\code{RECL}}(\adv) = \code{Pr}[\expred(\secuparam) = 1]$
\end{center}
\end{definition} 
\begin{theorem}
\label{thm:over-reclaim}
The system is secure against over-reclaiming providing that the smart card is a SE, the DS scheme is unforgeable, the ORAM scheme has integrity, and the commitment scheme is binding.
\end{theorem} 
\begin{proof}[Proof Sketch]
As we model the smart card as a SE, the adversary does not have access to the  \oraclemuregargs oracle. Hence, the adversary could win in three different ways:
(i)~it can forge a new transaction and the corresponding signature to create a valid proof for this transaction; 
(ii)~it can find another opening to the aggregated commitment scheme; or
(iii)~it can trick honest cards into creating signed transactions for an amount larger than the household budget.

In the first scenario, the adversary tries to create a new transaction with its corresponding proof
\begin{center}
    $\pi^{\code{(fake)}} = (\sigma^{\code{(fake)}}, \tau^{\code{(fake)}}, \code{Com}^{\code{(fake)}}, \code{r}^{\code{(fake)}})$
\end{center}

If the adversary can create such a proof, it means it has forged the signature $\sigma^{\code{(fake)}}$. This attack is thus infeasible as long as the chosen digital signature scheme is unforgeable.

In the second scenario, the adversary finds another opening to the commitment \code{Com}\textsubscript{\textsf{red}} so it opens to the amount the adversary desires. This attack is also impossible provided the chosen commitment scheme is binding. For the Pedersen commitment scheme that we use, this is the case under the DL assumption.

Finally, we consider the third scenario where the adversary tries to trick honest cards into creating transactions for a total amount that exceeds a household's budget. Cards always update the balance that is stored in the database for before outputting a transaction, so the only way in which cards will sign incorrect transactions is when the balance in the database is incorrect. In the case that the vendor can tamper with the database $\db$ it has two possible attack scenarios: modify the database, or roll-back the database to an earlier version. Since we assume the ORAM has integrity protection, modifying the database (and thus the balances) is not possible. When doing a roll-back however, the database will now contain a previous value of the household counter $\ctr$. Any new transaction with this rolled-back database will thus repeat a transaction tag $\tau$, and $\VerifyRedProof$ will return $\bot$.
\end{proof}

\subsection{Security of audit}
To ensure the correctness of the audit operation~(\reqlinky{auditsec}) we must show that even when vendors and the reclaim station collude, they cannot convince the auditor that the amount of money spent was higher than it was in reality.

We note that the reclaim station has no special powers whatsoever. It does not know any new key material and only verifies proofs provided by the vendors. As a result, any coalition of the vendor, reclaim station, and malicious users that can break auditability, would also be able to over-reclaim, violating Theorem~\ref{thm:over-reclaim}.

\subsection{Privacy at purchase -- unlinkability}
\label{sec:privacy-unlink}

\begin{algorithm}[t]
    \caption{Indistinguishability experiment}\label{alg:unlinkability-exp}
    \begin{algorithmic}[1]
      \footnotesize
      \Function{\expind}{\secuparam}\Comment{Malicious RS, malicious vendor}
        \State $(\skt, \pkt, \db) \leftarrow \TrustedSetup(\secuparam, N) $
        \State $(\sks, \pks) \leftarrow \SetupSign(\secuparam)$
        \State $\tidX{0}, \tidX{1}, \price^* , \varepsilon^*\leftarrow \adv^{\substack{\oraclemsreg,\\ \oraclespendmv}}(\sks, \pks)$
        \State $\pcif \tidX{0}, \tidX{1} \notin \myTH\; \pcreturn \perp$
        \State $\myst_0, \myst_1 \gets \SMap[\tidX{0}], \SMap[\tidX{1}]$
        \State $\{(\success_b, \price', \varepsilon')/\perp, - \}\leftarrow \langle \Spend(\myst_b, \price^*),  \adv\rangle$\label{my:line:firstChallenge}
        \State $\pcif (\price' \neq \price^* \vee \varepsilon' \neq\varepsilon^*) \pcthen \pcreturn \perp$
        \State $\adv\textsuperscript{\oraclemsreg, \oraclespendmv}(\sks, \pks)$\Comment{Cannot spend with any cards that were generated as copies of $\tidX{0}$ or $\tidX{1}$}
        \State $\{(\success_{1-b}, \price', \varepsilon')/\perp, - \} \leftarrow \langle \Spend(\myst_{1-b}, \price^*),  \adv\rangle$\label{my:line:secondChallenge}
        \State $\pcif (\price' \neq \price^* \vee \varepsilon' \neq\varepsilon^*) \pcthen \pcreturn \perp$
        \State $\pcif \success_b \oplus \success_{1-b} \pcthen \pcreturn \perp$
        \State $\myb' \leftarrow \adv\textsuperscript{\oraclemsreg, \oraclespendmv}(\sks, \pks)$
        \State $\pcreturn \myb' = \myb$
      \EndFunction
      \vspace{1mm}
    \end{algorithmic}
\end{algorithm}

A vendor should not be able to link a user to a specific household nor link two users as belonging to the same household~(\reqlinky{unlinkability}). We model this using the indistinguishability experiment in Algorithm~\ref{alg:unlinkability-exp}. Adversary \adv\xspace plays the role of a malicious vendor. We model the registration station as honest but curious.

Modeling indistinguishability in the wallet setting is tricky. We model a strong indistinguishability experiment where the adversary can interact with known cards of households before and after the challenge transaction. In standard indistinguishability experiments, the adversary would interact with \emph{one} of two challenge households. However, this approach does not work here: the challenge household would have less balance after the transaction (which the adversary could detect). Instead, we let the adversary interact with \emph{both} challenge households, and the adversary must determine the order.

 In the experiment of Algorithm~\ref{alg:unlinkability-exp} the adversary is given the keys of the registration station. It can then request honest users to register using $\oraclemsregargs(\cdot)$. It receives the full transcript. It can also ask honest users to spend money using $\oraclespendmvargs(\cdot)$, here the adversary plays the role of a malicious vendor. After interacting with the oracles, the adversary outputs the reclaim period and \price of a transaction as well as the identifiers $\tidX{0}, \tidX{1}$ of two different households it wants to try to distinguish (line 4). First, the challenger verifies that the chosen identifiers correspond to honest registered users (line 5). Then the adversary first interacts with $\tidX{b}$ (line 7). The challenger aborts if the transaction fails (line 8). Next, the adversary can again interact with different users (line 9), but not with any cards that are copies of the challenges households. After, the adversary interacts with the other challenge user $\tidX{1 - b}$ (lines 10--11). After both challenge transactions happened, the challenger checks whether the two households both succeeded or both failed their transactions (line 12). Otherwise, the adversary could trivially distinguish them using the output of the \Spend protocol. Finally, the adversary can once again interact with users (this time without any restrictions) before outputting its guess of the challenge bit $b'$ (line 13). The adversary wins if it guesses correctly.
\begin{definition}
The system provides unlinkability if the following advantage is negligible:
 \begin{center}
     $\code{Adv}^{\code{IND}}_{\adv} = |\code{Pr}[\code{Exp}^{\code{IND}}_{\adv, 0}(\secuparam)] - \code{Pr}[\code{Exp}^{\code{IND}}_{\adv, 1}(\secuparam)]|$ 
 \end{center}
\end{definition}
\begin{theorem}\label{th:unlink}
The system provides unlinkability if ORAM is oblivious, the PRF is a secure pseudorandom function, and the vendor does not rewind the database.
\end{theorem}

The rewinding assumption is essential. Since a household's cards cannot communicate, a malicious vendor that can control the database can always win the above game. Consider the following attack. It picks two households with the same budget, and two cards each. It executes a spend transaction for the full budget for each household during the challenge phase. After the challenge phase, it rewinds the database to the state between the two challenge transactions. And then tries to execute a transaction for the full amount with the first household using the so far unused card. Since cards cannot communicate, this card cannot know about any earlier transactions. If the transaction succeeded, the first challenge household was household two, otherwise it was household one. Thus breaking the indistinguishability game. We explain in Section~\ref{sec:disc} how these attacks can be detected. We now provide a proof sketch of Theorem~\ref{th:unlink}.

\begin{proof}[Proof Sketch]
We argue that $\expind[0]$ and $\expind[1]$ are indistinguishable to the adversary. First, notice that as a result of the trusted setup, and the fact that cards verify the key $\sks$ they receive, all cards share the same signing key $\sks$. Also, because the registration is honest, it does not craft extra cards for target households.

We focus on the two challenge transactions. First, notice that in both cases, the ORAM guarantees that the adversary cannot see which records are being read or written. Let $\pi^0 = (\sigma^0, \tau^0, \code{Com}^0, \code{r}^0)$ be the proof sent by $\tidX{0}$ and $\pi^1 = (\sigma^1, \tau^1, \code{Com}^1, \code{r}^1)$ the proof sent by $\tidX{1}$. Since the ORAM is secure,  any adversary that wins the game must distinguish cards based on the proofs $\pi^0$ and $\pi^1$. We now argue that the adversary cannot.

First, we argue that values $\tau^0$ and $\tau^1$ are indistinguishable from the adversary's point of view. By assumption, the adversary does not rewind the database / ORAM. Therefore, any transaction by the same household will always use a different value of $\ctr$. Hence the inputs to the PRF are different for all transactions of the challenge households. Furthermore, the because the adversary does not know the PRF key. We conclude that $\tau^0$ and $\tau^1$ are pseudorandom to the adversary.

Next, direct inspection shows that the the values $\code{r}^0$ and $\code{r}^1$ are independent of the household, and therefore so are the commitments $\code{Com}^0$ and $\code{Com}^1$. Finally, we conclude that  signatures 
$\sigma^0$ and $\sigma^1$ are independent because they are signatures over random inputs with the same signing keys. We conclude that a non-rewinding adversary cannot distinguish the challenge cards.
\end{proof}

\subsection{Privacy at reclaim}
The reclaim station should not be able to know the details of the transactions but only the total amount per vendor~(\reqlinky{reclaim}). As the proofs used for the reclaim phase and the auditing one are the same, privacy at reclaim holds as long as privacy at auditing does.

\subsection{Privacy at auditing}

\begin{algorithm}[t]
    \caption{Audit privacy experiment}\label{alg:audit-exp}
    \begin{algorithmic}[1]
      \footnotesize
      \Function{\oracleregtwoargs}{$\world, \bud, \sks, \tnb$}
        \State $\myst \leftarrow \SetupToken(\pks, \skt, \pkt)$
        \State $\{\myst / \perp, (\Hid, \db^{(\world)})/\db^{(\world)}\}\larrow \langle \Request(\myst), \Allocate(\sks, \bud, \db^{(\world)})\rangle$
        \State $\Bud^{(\world)}[\Hid] \leftarrow \bud$
        \State $\pcfor i \in \tnb$
        \State\hskip\algorithmicindent $\tid \leftarrow \counter^{(\world)} + 1$
        \State\hskip\algorithmicindent $\SMap^{(\world)}[\tid] \leftarrow \myst$
        \State\hskip\algorithmicindent $\code{NewIds} \leftarrow \code{NewIds} \cup \tid$
        \State $\myTH^{(\world)} \leftarrow \myTH^{(\world)} \cup \code{NewIds}$
        \State $\pcreturn \code{NewIds}$
      \EndFunction
      \vspace{1mm}

      \Function{\oraclespendtwoargs}{$\world, \varepsilon, \tid, \price$}
        \State $\myst \leftarrow \SMap^{(\world)}[\tid] \code{ or abort if does not exist}$
        \State $\{(\success, \price', \varepsilon')/\perp, (\pi, \db^{(\world)})/\perp\}\leftarrow\langle \Spend(\myst, \price), \Receive(\varepsilon, \price, \db^{(\world)})\rangle$
        \State $\pcif (\success \wedge \price = \price' \wedge \varepsilon = \varepsilon') \pcthen$
        \State \hskip\algorithmicindent $\inlog^{(\world)}[\varepsilon] \leftarrow \inlog^{(\world)}[\varepsilon] \cup (\price ,\pi)$ 
      \EndFunction
      \vspace{1mm}

      \Function{\expaud}{\secuparam
      } \Comment{Malicious Auditor}
        \State $(\skt, \pkt, \db) \leftarrow \TrustedSetup(\secuparam, N)$
        \State $(\sks, \pks) \leftarrow \SetupSign(\secuparam) $
        
        \State $\varepsilon^* \leftarrow \adv \textsuperscript{\oracleregtwo, \oraclespendtwo} (\pks)$
        \State $\spendsumz, \pi^0_{\code{aud}} \leftarrow \GenProof (\varepsilon^*, \inlog^0[\varepsilon^*])$
        \State $\spendsumo, \pi^1_{\code{aud}} \leftarrow \GenProof (\varepsilon^*, \inlog^1[\varepsilon^*])$
        \State $\pcif \spendsumz \neq \spendsumo \vee |\pi^0_{\code{aud}}| \neq |\pi^1_{\code{aud}}|
\pcreturn \perp$
        \State $\myb' \leftarrow \adv(\pi^\myb_{\code{aud}})$
        \State $\pcreturn \myb' = \myb$
      \EndFunction
      \vspace{1mm}
    \end{algorithmic}
\end{algorithm}

The auditor should not be able to know the details of the transactions but only the total amount per vendor~(\reqlinky{audit}). We follow Wang et al.~\cite{wang2023digital} and model this using the auditability experiment from Algorithm~\ref{alg:audit-exp} inspired by Benaloh's ballot privacy game \cite{bernhard2015ballot}. 
We want to model that the auditor learns nothing more than the total amount of money spent, and the number of transactions. The key idea of the game is to let the adversary build up two sets of transactions, we call these sets \emph{worlds}, where recipients use their cards to spend their budget with a vendor. Contrary to Wang et al., the adversary is fully free to register parties (using $\oracleregtwoargs(\cdot)$) in either world, and to subsequently ask specific cards of a household to make a spend transaction (using \oraclespendtwo).

After interacting with the oracles, the adversary outputs a challenge reclaim period $\varepsilon^*$. First, the challenger computes the sum of transactions and the auditable proof for the chosen period in each world. If the sums or the length of the proofs are different the game aborts as the adversary can trivially distinguish them. Finally, the challenger chooses one of the two worlds and gives the corresponding proof to the adversary. This one has to guess which world the challenger chose. It wins if it guesses correctly.
\begin{definition}
The system provides audit privacy if the following advantage is negligible:
 \begin{center}
     $\code{Adv}^{\code{AUDP}}_{\adv} = |\code{Pr}[\code{Exp}^{\code{AUDP}}_{\adv, 0}(\secuparam)] - \code{Pr}[\code{Exp}^{\code{AUDP}}_{\adv, 1}(\secuparam)]|$ 
 \end{center}
\end{definition}
\begin{theorem}
The system provides privacy if the commitment scheme is hiding, and the PRF is secure.
\end{theorem}
\begin{proof}[Proof Sketch]
In this proof sketch we focus on showing that the lists of transactions provided to the adversary are indistinguishable.
Recall that a transaction proof is of the form $\pi^b_{\code{aud}} = (\code{r}_\code{sum}, (\sigma^{(i)}, \tau^{(i)}, \code{Com}^{(i)})_{i})$, where each $\sigma^{(i)}$ is a signature over $\tau^{(i)} \parallel \varepsilon^* \parallel \Com^{(i)}$. We argue that these are all independent of the choice of world $b$.

First, by construction, cards in either world are registered with the same key $\sks$, hence the signatures $\sigma^{(i)}$ themselves do not give information about $b$ as long as the message it signs does not. The commitment $\Com^{(i)}$ is perfectly hiding, and $\varepsilon^{*}$ is the same in both worlds. This leaves the PRF outputs $\tau^{(i)}$. Note that each input $\Hid \parallel \ctr$ is used at most once in a single world; hence because the PRF is secure and the adversary does not know the key, the $\tau^{(i)}$ are pseudorandom as well.

To complete the proof requires showing that the above still holds even though we reveal $\code{r}_\code{sum}$ as well. However, the argument that this works as long as the commitment scheme is hiding is exactly the same as in Wang et al.~\cite{wang2023digital}, so we omit it here.
\end{proof}

%% file: parts/6-discussion.tex
Here, we discuss practical considerations, extensions, and real-world challenges. 

\descr{Protecting against rollback attacks.} As we explained in Section~\ref{sec:privacy-unlink}, a fully malicious vendor with control over the database could roll-back the database to an earlier state and use this to distinguish users. We cannot guarantee that cards can always detect this attack, but we can make a simple modification that ensures that misbehaving vendors will likely be detected. In reality, vendors are unlikely to misbehave if that would lead them to be caught: the cost would be too high. The humanitarian organization can exclude them from the system going forward. Therefore we make the assumption that vendors are \emph{covert}.

The protection is simple. Every card locally keeps track of the last value of $\ctr$ that they wrote to the ORAM. Whenever they receive an older value (indicating a roll-back) they record the violation and stop working. Since all cards are identical from the vendor's perspective up to the point of the ORAM read operation, a misbehaving vendor cannot know whether it is dealing with a card that could potentially detect a rollback or another card of the same household that cannot. If the vendor performs the roll-back anyway, it risks being caught. Therefore a covert adversary will never perform a roll-back. (Incidentally, the solution by Wang et al.~\cite{wang2023digital} implicitly makes similar assumptions to avoid roll-backs of the block list.)

\descr{Periodicity.} Requirement \reqlinky{period} calls for periodical budget allocations: i.e., the budget is renewed after each period. The precise details depend on the requirements of the humanitarian organizations, e.g., whether unspent budget expires, or is carried over to the next round. In either case, we can make a simple modification to our algorithms to support these cases. First, we modify the cards to locally track what is the periodic allowance and when this allowance should be added to the current balance (or topped up). Next, we modify the stored record in the database to include when the balance was last updated.

Every time a card retrieves the current balance it first checks if the balance should be updated (e.g., because we now entered a new period). If so, it will apply the update rule (either add the periodic allowance, or fill-up the balance to the maximum) before proceeding. When writing the new balance it will also update the last-update value so that other cards of the same household will not also add new balance.

\descr{Vendor(s) requirements and the limits of several vendors.} 
For deploying our solution, the hardware requirement for vendors is limited to a card reader and a low-powered computer (\eg a Raspberry Pi), and, for the multi-vendor setting, network connectivity.

We note that in the multi-vendors setting, the main challenge is the synchronization of the database. In our presentation, we have assumed that all vendors and the registration station interact with the same database. This requires connectivity between them, or between these parties and an external (untrusted) server that hosts the database (a Raspberry Pi would also suffice for hosting this database). The lack of support for concurrent access in our protocols limits the level to which we can scale this approach without having to reduce privacy by creating different ORAM shards.

Communication between vendors, however, is essential. By definition, cards of the same household cannot communicate themselves, so we need an external mechanism to synchronize their state. If there are multiple vendors, these vendors must then interact among each other, lest a household can double-spend their budget at several vendors. 

\descr{Alternative design to replace reclaiming phase.} In this work, we followed the architecture of Wang \etal~\cite{wang2023digital} to let vendors prove how much money they can reclaim based on individual transactions. This approach provides flexibility. Vendors can retroactively decide on the granularity of the proofs (e.g., a day, week, or month); and even recompute these if needed. The downside is that this requires the use of homomorphic commitments.

If such flexibility is not needed, a simpler approach is possible. Since the smart cards are considered an SE, and thus we can trust the execution to always be honest, we could instead use them to update a per-vendor balance.\footnote{We thank the anonymous reviewer for suggesting this approach.} The balance record consists of the balance, a per-period nonce, and a card signature over the first two. For each transaction, the card retrieves the latest signed vendor balance, verifies the signature, updates the balance with the latest transaction amount, signs the new balance and the current per-period nonce, and sends it back to the vendor. The vendor can reclaim what it is owed by simply sending the signed balance and per-period nonce to the reclaim station. The reclaim station verifies the signature and that the nonce is fresh before reimbursing the vendor. The use of the nonce and the freshness check ensures that a vendor cannot overclaim by continuing a previous period. To start a new period, the vendor sets the balance to zero. The first card to encounter a zero balance will not verify the missing signature, create a new nonce, and then proceed as before to return a signed balance with the new nonce.

%% file: parts/5-evaluation.tex
We evaluate our system to show that it scales to reasonable numbers. We focus on the balance retrieval and update part of the transaction phase as this is the most costly aspect of our design. We evaluate the performance cost on a NXP J3H145 dual
144k Java Card smart card (the same one that was used in prior work~\cite{wang2023digital}). Our code is open source and can be found here: \url{https://github.com/wouterl/humanitarian-wallet-code}.

We split the analysis of the cost into two parts: (1) the cost of the ORAM access which we evaluate using our new implementation and (2) the cost of all the other computations (PRF, signature, commitments) for which we use prior numbers~\cite{wang2023digital}.

\descr{The cost of ORAM accesses.} We first evaluate the cost of retrieving the latest balance and counter from the ORAM, and then writing the updated values back. We use 16 bits for the counter and 16 bits for the balance. We store these in a single 4-byte field. Throughout this section, we show mean transaction times based on 10 runs and compute the standard error of the mean (but in all cases the standard error was too small to be visible in the graphs and is thus not shown).

We evaluate the cost of two ORAMs. The first is the \emph{naive ORAM} where the database is a single array of 4-byte values encrypted with authenticated encryption. Cards simply download the full ciphertext, decrypt the household's balance at the household's position, modify it, and upload a new fresh encryption. See Appendix~\ref{ap:oram} for the details. Since our smart card does not natively support authenticated encryption, we implement it using the encrypt-then-MAC paradigm using AES in CBC mode for encryption and CMAC with AES for authentication.

Figure~\ref{fig:naive-perf} shows the performance of our proof-of-concept implementation on the NXP J3H145 dual 144k Java Card. We measured end-to-end time as well as, separately, the time it takes to send and retrieve data from the card when it does not perform any further processing. We observe that it takes less than $20$s to perform a transaction for just over $8,000$ households. This cost increases linearly with the number of households.

Due to the linear cost of naive ORAM, at some point, downloading and processing the full database becomes prohibitively costly. We therefore implemented and evaluated a \emph{recursive tree-based ORAM}~\cite{shi2011treeoram}. Since we access this ORAM using different cards, we must store the position map in the ORAM as well. To get integrity, we use authenticated encryption at each node of the tree, see Appendix~\ref{ap:oram} for the details. Again we used AES in CBC mode for encryption and CMAC with AES for authentication. Figure~\ref{fig:tree-perf} shows the total cost as well as the transfer cost. We observe that both runtime and communication overheads are sub-linear in the number of households. 

Figure~\ref{fig:comparison-perf} compares the performances of one transaction for both ORAM instantiations. We observe that the recursive tree-based instantiation becomes more efficient starting with 35k households, and then scales much better. Part of this cost-saving is due to a reduction in download cost. At $2^{15}$ households, the recursive ORAM has a total transfer size of only 114\,kB compared to 287\,kB for the naive ORAM. These transfers induce a significant time cost, see Figure~\ref{fig:comparison-transfer}.

\begin{figure*}
\centering
\begin{subfigure}[t]{0.49\textwidth}
    \centering
    \includegraphics[scale=0.5]{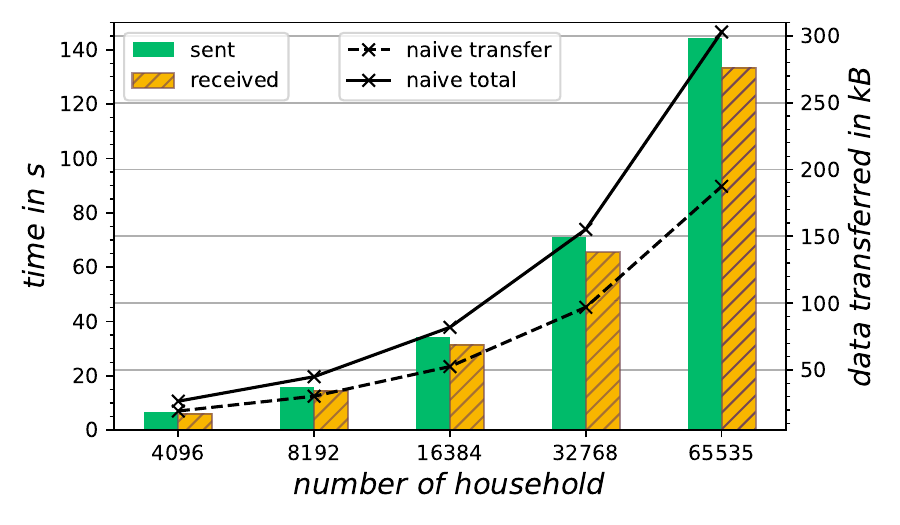}
    \caption{Naive ORAM costs per number of households.}
    \label{fig:naive-perf}
\end{subfigure}
\begin{subfigure}[t]{0.49\textwidth}
    \centering
    \includegraphics[scale=0.5]{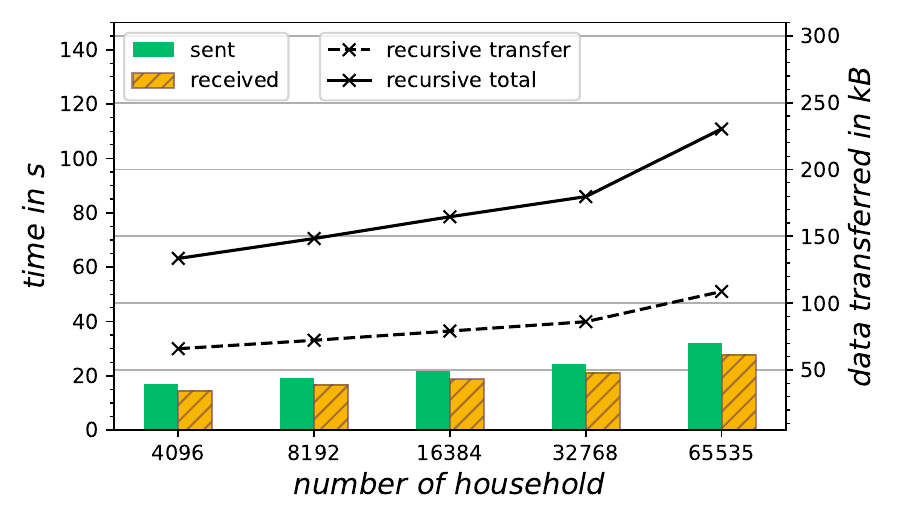}
    \caption{Recursive tree-based ORAM costs per number of households.}
    \label{fig:tree-perf}
\end{subfigure}
\caption{Evaluation of ORAM costs per transaction. The plain line corresponds to the total ORAM execution runtime (in s) while the dashed line corresponds to the transfer time only. The bar plots refer to the measure of the communication costs (sent/received by the card) in kB.}
\label{fig:evaluation}
\end{figure*}

\begin{figure*}
\centering
\begin{subfigure}[t]{0.49\textwidth}
    \centering
    \includegraphics[scale=0.5]{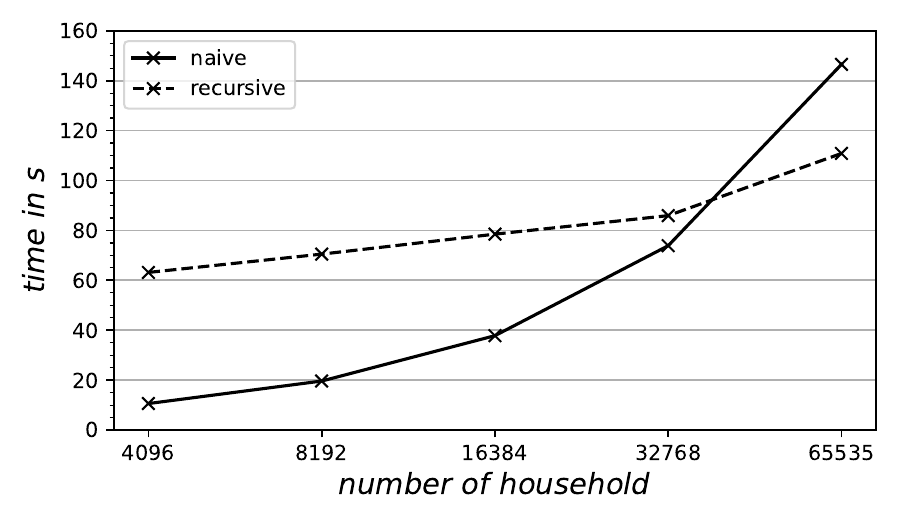}
    \caption{Total costs: Naive vs. recursive tree-based ORAM.}
    \label{fig:comparison-perf}
\end{subfigure}
\begin{subfigure}[t]{0.49\textwidth}
    \centering
    \includegraphics[scale=0.5]{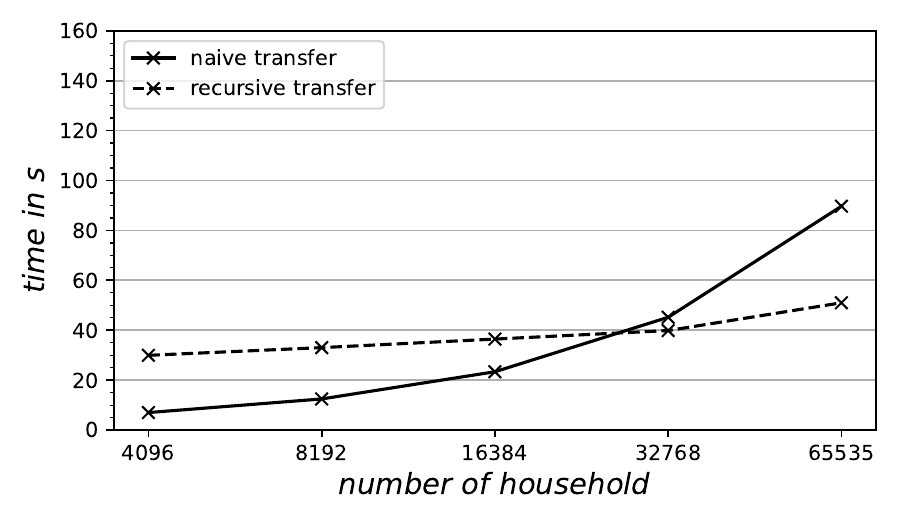}
    \caption{Transfer costs: Naive vs. recursive tree-based ORAM.}
    \label{fig:comparison-transfer}
\end{subfigure}
        
\caption{Comparison of ORAM costs per transaction. The plain line corresponds to the naive ORAM while the dashed line corresponds to the recursive tree-based ORAM. }
\label{fig:comparison}
\end{figure*}

\descr{The other computations.}
Finally, we need to add the cost of the other computations.
As in Wang et al.~\cite{wang2023digital} we use ECDSA (P256 curve) for the digital signature scheme and use a slightly tweaked Pedersen's commitment~\cite{wang2023digital}.
From prior work~\cite{wang2023digital} we know that it takes 500\,ms for computing the Pedersen commitment, 166\,ms for the signature, and 113\,ms for computing the PRF. To obtain end-to-end cost, one must add around 800\,ms to the results from Figure~\ref{fig:evaluation}.

\descr{Transfer cost.}
The graphs in Figure~\ref{fig:evaluation} also show the transfer cost in bytes between the card and the card reader. These numbers include the APDU headers, thus explaining why more data is sent to the card than received from it in Figure~\ref{fig:naive-perf}.
The numbers shown in Figure~\ref{fig:evaluation} also give an upper-bound on the amount of data that must be transmitted by vendors when using a central database. Even with 128k households, this amounts to less than 150\,kB when using recursive ORAM.

\descr{Practicality.} These numbers show that for small distribution programs of up to a few thousand households, our approach is practical: the total cost of a single transaction is under 10 seconds for up to 2,000 households. For larger distribution programs the transaction time remains under 1 minute, but we would need faster cards to make these sizes practical.

We also observe that the communication time is high compared to the expected transfer speeds quoted for the cards we use. We measured an effective speed of 5--10 kBps, whereas based on the datasheets we would have expected an order of magnitude more. A better implementation (or better hardware) that realizes the expected speed would yield 2-3$\times$ faster end-to-end times.

%% file: parts/7-conclusion.tex
Humanitarian organizations are increasingly interested in digitalization to gain efficiency in their processes. To enable them to do so it is important to develop digital solutions that are tailored to their needs, often different than the requirements of our daily life applications.

In this work, we designed a system to digitize humanitarian cash distribution workflow in collaboration with the \ICRC. Our new auditable digital wallet scheme enables the \ICRC\ to maintain the privacy properties enabled by previous work on aid distribution, while allocating a flexible budget to aid recipients such that they can obtain the items they need the most at every point in time. 
Our design uses smart cards as tokens and relies on different flavors of an ORAM protocol to ensure versatility in the field. We conducted a thorough security and privacy analysis and evaluated the performance of our solution showing its practicality at the scale needed in the humanitarian sector.

%% file: parts/appendix.tex
\section{ORAM}\label{ap:oram}
ORAM protocols enable a client to hide her access pattern to an external memory thus helping to guarantee privacy. In our system, the smart card uses ORAM to find its current balance in the vendor's database without disclosing its access pattern. This way, the smart card does not have to reveal to which household it belongs. Depending on the size of the community, we have to adapt the way we instantiate \ORAM.

\descr{Authenticated Encryption.}
Let \code{AE} be an authenticated encryption scheme s.t. $\code{AE}\xspace {:=} (\KGen, \Enc, \Dec)$. The key generation algorithm $k \gets \AEnc.\KGen(\secuparam)$ takes as input a security parameter $\ell$ and outputs a secret key $k$. The encryption algorithm $\cipher {\larrow} \AEnc.\Enc(k, \mess)$ takes as input the key $k$ and a message $\mess {\in} \{0, 1\}^*$ and outputs a ciphertext $\cipher$. The decryption algorithm $\mess {\larrow} \AEnc.\Dec(k, \cipher)$ returns a message $\mess$ or an error $\perp$.

\descr{Naive ORAM for small community}
We now describe the \textit{naive ORAM} protocol. The idea of this protocol is to store the database in encrypted form on the server. The client downloads all of it, decrypts it, reads its record, updates it, re-encrypts the database, and sends it back to the server. We instantiate \ORAM as follows:
\begin{itemize}
    \item (\sk, \db) \larrow \ORAM.\Init(\secuparam, \N). The trusted party sets up the authenticated encryption scheme by running $k \gets \AEnc.\KGen(\secuparam)$, and sets $\skt=k$. Then, it initializes an array $\code{data}$ of \N empty elements, (\N being the number of households) and encrypts it into ciphertext $\cipher$ using $\AEnc.\Enc(\skt, \code{data})$. Let $\db = \cipher$. It returns (\sk, \db).
    \item \db \larrow \ORAM.\Serve(\db). The server sends the encrypted database \db to the client. Then, the server waits until the client sends back the updated database and stores it.
    \item \code{data} \larrow \ORAM.\Read(\skt, \Hid). The client waits to receive the encrypted database \db. She decrypts the \db using $\code{data} = \AEnc.\Dec(\skt, \db)$, and any tampering to the database results in the decryption to return $\bot$. The client reads the record $\code{data}[\Hid]$ stored in the $\Hid$-th component of the array $\code{data}$. She re-encrypts the database as $\db' = \AEnc.\Enc(\skt, \code{data})$. Finally, she sends back the encrypted $\db'$ to the server.
    \item $\perp \larrow \ORAM.\Write(\skt, \Hid, \code{record})$. The client waits to receive the database $\db$. She decrypts it using $\code{data} = \AEnc.\Dec(\skt, \db)$. The client then writes her record by setting $\code{data}[\Hid] = \code{record}$ into the $\Hid$-th position of the plaintext array $\code{data}$ and re-encrypts it into $\db' = \AEnc.\Enc(\skt, \code{data})$. Finally, she sends back the $\db'$ to the server.
\end{itemize}
We employ this version of ORAM in the case of a small community. The client is the registered user's smart card and the registration station and the vendor play the server's role.

\descr{Tree-based (recursive) ORAM for bigger communities}
When the database is big, the naive ORAM is surpassed by more elaborated constructions that avoid retrieving the whole database. In our evaluation, we experimented with tree-based ORAM protocol~\cite{shi2011treeoram, stefanov2018pathoram} and recursive versions of these. The challenge with these efficient ORAM constructions is that to hide which records are read, while not requiring reading the entire database, they must move to different places in the underlying database lest repeated queries leak. Therefore tree-based ORAMs use a position map that tells clients where the blocks are, however, the position map itself is rather large, and cannot be stored by the clients in our case. The solution that we adopted is to use recursive ORAM, where the position map is itself stored in a sequence of smaller ORAMs.

\section{Creating multiple cards for the same household}
\label{append:whisper}
To provide robustness (\reqlinky{robust}) to our protocol, we provide each household with several smart cards. According to our previously described protocol~\ref{sec:sol}, we need the smart cards of the same household to share the same state. In our instantiation, there is no shared state between cards of the same household other than the household ID. Therefore, since the registration station is honest, it simply initiates the other cards as well, sending them the household ID. All these cards will write the same balance to the ORAM, but for efficiency the extra cards could omit this step.

%% file: icrc-wallet.bbl
\begin{thebibliography}{10}
\providecommand{\url}[1]{#1}
\csname url@samestyle\endcsname
\providecommand{\newblock}{\relax}
\providecommand{\bibinfo}[2]{#2}
\providecommand{\BIBentrySTDinterwordspacing}{\spaceskip=0pt\relax}
\providecommand{\BIBentryALTinterwordstretchfactor}{4}
\providecommand{\BIBentryALTinterwordspacing}{\spaceskip=\fontdimen2\font plus
\BIBentryALTinterwordstretchfactor\fontdimen3\font minus
  \fontdimen4\font\relax}
\providecommand{\BIBforeignlanguage}[2]{{%
\expandafter\ifx\csname l@#1\endcsname\relax
\typeout{** WARNING: IEEEtranS.bst: No hyphenation pattern has been}%
\typeout{** loaded for the language `#1'. Using the pattern for}%
\typeout{** the default language instead.}%
\else
\language=\csname l@#1\endcsname
\fi
#2}}
\providecommand{\BIBdecl}{\relax}
\BIBdecl

\bibitem{AlmashaqbehS22SoKblockchain}
G.~Almashaqbeh and R.~Solomon, ``{SoK}: Privacy-preserving computing in the
  blockchain era,'' in \emph{{IEEE} European Symposium on Security and Privacy
  (EuroS{\&}P)}, 2022, \url{https://doi.org/10.1109/EuroSP53844.2022.00016}.

\bibitem{DBLP:conf/pkc/BaldimtsiCFK15}
F.~Baldimtsi, M.~Chase, G.~Fuchsbauer, and M.~Kohlweiss, ``Anonymous
  transferable e-cash,'' in \emph{Public-Key Cryptography ({PKC})}, 2015,
  \url{https://doi.org/10.1007/978-3-662-46447-2\_5}.

\bibitem{DBLP:conf/sp/Ben-SassonCG0MTV14}
E.~Ben{-}Sasson, A.~Chiesa, C.~Garman, M.~Green, I.~Miers, E.~Tromer, and
  M.~Virza, ``Zerocash: Decentralized anonymous payments from bitcoin,'' in
  \emph{{IEEE} Symposium on Security and Privacy ({SP})}, 2014,
  \url{https://doi.org/10.1109/SP.2014.36}.

\bibitem{bernhard2015ballot}
D.~Bernhard, V.~Cortier, D.~Galindo, O.~Pereira, and B.~Warinschi, ``{SoK}: {A}
  comprehensive analysis of game-based ballot privacy definitions,'' in
  \emph{{IEEE} Symposium on Security and Privacy ({SP})}, 2015,
  \url{https://doi.org/10.1109/SP.2015.37}.

\bibitem{BurtonICRC20}
J.~Burton, ``{“Doing no harm” in the digital age: What the digitalization
  of cash means for humanitarian action},''
  \url{https://international-review.icrc.org/sites/default/files/reviews-pdf/2021-03/doing-no-harm-digitalization-of-cash-humanitarian-action-913.pdf},
  2020, accessed: June 6th, 2024.

\bibitem{camenisch2005compact}
J.~Camenisch, S.~Hohenberger, and A.~Lysyanskaya, ``Compact e-cash,'' in
  \emph{Advances in Cryptology - {EUROCRYPT}}, 2005,
  \url{https://doi.org/10.1007/11426639\_18}.

\bibitem{camenisch2006balancing}
------, ``Balancing accountability and privacy using e-cash (extended
  abstract),'' in \emph{Security and Cryptography for Networks (SCN)}, 2006,
  \url{https://doi.org/10.1007/11832072\_10}.

\bibitem{DBLP:conf/acns/CanardG08}
S.~Canard and A.~Gouget, ``Anonymity in transferable e-cash,'' in \emph{Applied
  Cryptography and Network Security ({ACNS})}, 2008,
  \url{https://doi.org/10.1007/978-3-540-68914-0\_13}.

\bibitem{DBLP:conf/pkc/CanardPST15}
S.~Canard, D.~Pointcheval, O.~Sanders, and J.~Traor{\'{e}}, ``Divisible e-cash
  made practical,'' in \emph{Public-Key Cryptography - ({PKC})}, 2015,
  \url{https://doi.org/10.1007/978-3-662-46447-2\_4}.

\bibitem{Chaum82ecash}
D.~Chaum, ``Blind signatures for untraceable payments,'' in \emph{Advances in
  Cryptology - {CRYPTO}}, 1982,
  \url{https://doi.org/10.1007/978-1-4757-0602-4\_18}.

\bibitem{ChaumFN88ecash}
D.~Chaum, A.~Fiat, and M.~Naor, ``Untraceable electronic cash,'' in
  \emph{Advances in Cryptology - {CRYPTO}}, 1988,
  \url{https://doi.org/10.1007/0-387-34799-2\_25}.

\bibitem{ciesielski2022afghans}
R.~Ciesielski and M.~Zierer, ``How biometric devices are putting afghans in
  danger,'' 2022,
  \url{https://interaktiv.br.de/biometrie-afghanistan/en/index.html}.

\bibitem{DBLP:conf/ndss/DanezisM16}
G.~Danezis and S.~Meiklejohn, ``Centrally banked cryptocurrencies,'' in
  \emph{Network and Distributed System Security Symposium ({NDSS})}, 2016,
  \url{http://wp.internetsociety.org/ndss/wp-content/uploads/sites/25/2017/09/centrally-banked-cryptocurrencies.pdf}.

\bibitem{goldie2022partisia}
A.~Goldie, ``Partisia blockchain partners with the international committee of
  the {Red Cross} to address real-world challenges in conflict zones,'' 2022,
  \url{https://medium.com/partisia-blockchain/partisia-blockchain-partners-with-the-international-committee-of-the-red-cross-to-address-af05d4a213fb}.

\bibitem{Goldreich87oram}
O.~Goldreich, ``Towards a theory of software protection and simulation by
  oblivious {RAMs},'' in \emph{{ACM} Symposium on Theory of Computing
  ({STOC})}, 1987, \url{https://doi.org/10.1145/28395.28416}.

\bibitem{icrc2019ecosec}
{ICRC}, ``Ecosec response,'' 2019,
  \url{https://shop.icrc.org/ecosec-response-en-pdf.html}.

\bibitem{icrc2023partisia}
------, ``Humanitarian token solution: Digital cash assistance that preserves
  privacy,'' 2023,
  \url{https://blogs.icrc.org/inspired/2023/06/27/humanitarian-token-solution-digital-cash-assistance-preserves-privacy/}.

\bibitem{ifrc2021identity}
{IFRC (International Federation of Red Cross and Red Crescent Societies},
  ``Digital identity: enabling dignifies access to humanitarian services in
  migration,'' 2021,
  \url{https://preparecenter.org/wp-content/uploads/2021/06/Digital-Identity-Enabling-dignified-access-to-humanitarian-services-in-Migration-Final.pdf}.

\bibitem{KaspersenICRC16}
A.~Kaspersen and C.~Lindsey-Curtet, ``{The digital transformation of the
  humanitarian sector},''
  \url{https://blogs.icrc.org/law-and-policy/2016/12/05/digital-transformation-humanitarian-sector/},
  2016, accessed: June 6th, 2024.

\bibitem{DBLP:conf/ccs/KiayiasKS22}
A.~Kiayias, M.~Kohlweiss, and A.~Sarencheh, ``{PEReDi}: Privacy-enhanced,
  regulated and distributed central bank digital currencies,'' in \emph{{ACM}
  {SIGSAC} Conference on Computer and Communications Security ({CCS})}, 2022,
  \url{https://doi.org/10.1145/3548606.3560707}.

\bibitem{MiersG0R13}
I.~Miers, C.~Garman, M.~Green, and A.~D. Rubin, ``Zerocoin: Anonymous
  distributed {E-Cash} from bitcoin,'' in \emph{{IEEE} Symposium on Security
  and Privacy ({SP})}, 2013, \url{https://doi.org/10.1109/SP.2013.34}.

\bibitem{aljazeera2019kenya}
O.~Osman, ``Kenyan somali refugees claim they are denied citizenship rights,''
  Al Jazeera, 2019,
  \url{https://www.aljazeera.com/features/2019/5/19/kenyan-somali-refugees-claim-they-are-denied-citizenship-rights}.

\bibitem{pedersen1991commit}
T.~P. Pedersen, ``Non-interactive and information-theoretic secure verifiable
  secret sharing,'' in \emph{Advances in Cryptology - {CRYPTO}}, 1991,
  \url{https://doi.org/10.1007/3-540-46766-1\_9}.

\bibitem{RinbergA22SoKanonymity}
R.~Rinberg and N.~Agarwal, ``Privacy when everyone is watching: An {SOK} on
  anonymity on the blockchain,'' \emph{{IACR} Cryptol. ePrint Arch.}, 2022,
  \url{https://eprint.iacr.org/2022/985}.

\bibitem{ShepherdAGLMASC16}
C.~Shepherd, G.~Arfaoui, I.~Gurulian, R.~P. Lee, K.~Markantonakis, R.~N. Akram,
  D.~Sauveron, and E.~Conchon, ``Secure and trusted execution: Past, present,
  and future - {A} critical review in the context of the internet of things and
  cyber-physical systems,'' in \emph{{IEEE} Trustcom/BigDataSE/ISPA}, 2016,
  \url{https://doi.org/10.1109/TrustCom.2016.0060}.

\bibitem{shi2011treeoram}
E.~Shi, T.~H. Chan, E.~Stefanov, and M.~Li, ``Oblivious ram with $o((log n)^3)$
  worst-case cost,'' \emph{{IACR} Cryptol. ePrint Arch.}, 2011,
  \url{http://eprint.iacr.org/2011/407}.

\bibitem{stefanov2018pathoram}
E.~Stefanov, M.~van Dijk, E.~Shi, T.~H. Chan, C.~W. Fletcher, L.~Ren, X.~Yu,
  and S.~Devadas, ``Path {ORAM:} an extremely simple oblivious {RAM}
  protocol,'' \emph{J. {ACM}}, vol.~65, no.~4, 2018,
  \url{https://doi.org/10.1145/3177872}.

\bibitem{UngerDBFPG015SoKmessaging}
N.~Unger, S.~Dechand, J.~Bonneau, S.~Fahl, H.~Perl, I.~Goldberg, and M.~Smith,
  ``{SoK}: Secure messaging,'' in \emph{{IEEE} Symposium on Security and
  Privacy ({SP})}, 2015, \url{https://doi.org/10.1109/SP.2015.22}.

\bibitem{wang2023digital}
B.~Wang, W.~Lueks, J.~Sukaitis, V.~G. Narbel, and C.~Troncoso, ``Not yet
  another digital {ID:} privacy-preserving humanitarian aid distribution,'' in
  \emph{{IEEE} Symposium on Security and Privacy ({SP})}, 2023,
  \url{https://doi.org/10.1109/SP46215.2023.10179306}.

\bibitem{wfp2023buildingblocks}
{World Food Program}, ``Building blocks blockchain network for humanitarian
  assistance - graduated project,'' 2023,
  \url{https://innovation.wfp.org/project/building-blocks}.

\bibitem{DBLP:conf/fc/WustKCC19}
K.~W{\"{u}}st, K.~Kostiainen, V.~Capkun, and S.~Capkun, ``{PRCash}: Fast,
  private and regulated transactions for digital currencies,'' in
  \emph{Financial Cryptography and Data Security ({FC})}, 2019,
  \url{https://doi.org/10.1007/978-3-030-32101-7\_11}.

\bibitem{DBLP:conf/ccs/WustKDC22}
K.~W{\"{u}}st, K.~Kostiainen, N.~Delius, and S.~Capkun, ``Platypus: {A} central
  bank digital currency with unlinkable transactions and privacy-preserving
  regulation,'' in \emph{{ACM} {SIGSAC} Conference on Computer and
  Communications Security ({CCS})}, 2022,
  \url{https://doi.org/10.1145/3548606.3560617}.

\end{thebibliography}
